\documentclass[pdflatex,sn-mathphys-num]{sn-jnl}

\usepackage{amsmath}
\usepackage{amssymb}
\usepackage{mathtools}
\usepackage{bm}
\usepackage{physics}
\usepackage{graphicx}
\usepackage{hyperref}

\theoremstyle{thmstyleone}
\newtheorem{theorem}{Theorem}[section]
\newtheorem{proposition}[theorem]{Proposition}

\theoremstyle{thmstyletwo}

\theoremstyle{thmstylethree}
\newtheorem{definition}[theorem]{Definition}

\begin{document}
	
	\title[
	Fuchsian Resonance and Pole-Skipping
	]
	{Fuchsian Resonance and a Horizon-to-Boundary Dictionary for Pole-Skipping}
	
	
	\author*[1]{\fnm{Yoon-Seok} \sur{Choun}}
	\email{ychoun@gmail.com}

	\affil*[1]{
		\orgdiv{Department of Physics},
		\orgname{POSTECH},
		\orgaddress{ \city{ Pohang, Gyeongbuk 37673},
			\country{Republic of Korea}	}
	} 
	
	\abstract{

Pole-skipping occurs at special complex frequencies and momenta where the retarded Green function of a black-hole or black-brane background is not uniquely defined. The local near-horizon mechanism is well known: at resonance, a Frobenius recurrence matrix loses rank and the space of smooth horizon solutions enlarges. We address the global question of how two boundary-normalized solutions enter this resonant horizon solution space.

For a general second-order scalar radial equation on a nonextremal background analytic near the horizon, the Frobenius recurrence at the resonant order yields a solvability condition. We prove that its vanishing is equivalent to singularity of the horizon recurrence matrix, absence of the logarithmic Frobenius term, and existence of two independent smooth horizon solutions. Away from resonance, the source zero is equivalent to horizon smoothness of the response-normalized solution, while the response zero is equivalent to horizon smoothness of the source-normalized solution.

We then analyze the parameter-dependent simple pole of the nonresonant ingoing solution as resonance is approached. After removing this singularity, its resonant limit is proportional to the larger-root Frobenius solution, with proportionality factor given by the same solvability function. Continuing to the boundary shows that the same condition is equivalent to simultaneous vanishing of the two regularized boundary connection coefficients. This establishes a horizon-to-boundary dictionary between local Fuchsian resonance and the boundary source-response 0/0 structure of pole-skipping.

Finally, different directions of approach can select different resonant solutions, with the selection governed by the first parameter variation of the same solvability function.
	}
	
	
	\maketitle
	
\section{Introduction}
\label{sec:introduction}

\subsection{Pole-skipping and the known local horizon structure}
\label{subsec:intro-local}

In real-time holography, the retarded Green function is computed by imposing an ingoing boundary condition at the future black-hole horizon.
At generic frequencies and momenta, this condition determines the bulk solution uniquely up to an overall normalization, and therefore also fixes the ratio between the source and response at the boundary
\cite{SonStarinets2002,HerzogSon2003,SkenderisVanRees2008,VanRees2009,KovtunStarinets2005}.
Let $H_{\rm in}$ denote the global solution obtained by selecting the ingoing Frobenius branch near the horizon and continuing it to the boundary.
Near the boundary, let $R_1$ and $R_2$ be two independent solutions normalized with respect to the two independent asymptotic channels.
The boundary expansion of $H_{\rm in}$ then takes the form
\begin{equation}
	H_{\rm in}
	=
	A R_1+B R_2.
\end{equation}
In the standard holographic setting, $R_1$ and $R_2$ may be chosen as the source-normalized and response-normalized solutions, respectively.
This source--response identification is the standard AdS/CFT boundary dictionary
\cite{Maldacena1998,Witten1998,GubserKlebanovPolyakov1998}.
With this convention, the retarded Green function is, up to an overall normalization,
\begin{equation}
	G_R\propto\frac{B}{A},
\end{equation}
as in Refs.~\cite{SonStarinets2002,HerzogSon2003,KovtunStarinets2005}.
The special complex frequency--momentum structure associated with pole-skipping first appeared in the context of black-hole scrambling and holographic energy dynamics
\cite{GrozdanovSchalmScopelliti2018,BlakeDavisonGrozdanovLiu2018},
and was subsequently formulated in terms of the nonuniqueness of the ingoing solution and the intersection of poles and zeros of generic holographic Green functions
\cite{BlakeDavisonVegh2020,NatsuumeOkamura2019,NatsuumeOkamura2020}.
At the boundary, the source and response coefficients may then vanish simultaneously, producing a $0/0$ Green function whose limiting value can depend on the direction of approach in parameter space.
Pole-skipping has since been studied in the presence of finite-coupling and higher-curvature corrections, at zero temperature, for nonmaximal chaos, for fermionic and Rarita--Schwinger fields, in general bosonic and gauge systems, for hyperbolic black holes, in Lifshitz and AdS$_2$ geometries, and in two-dimensional gravity and CFT
\cite{NatsuumeOkamuraFinite2019,WuHigherCurvature2019,NatsuumeOkamuraZeroTemp2021,
ChoiMezeiSarosi2021,CeplakRamdialVegh2020,CeplakVegh2021,WangWang2022,
NingWangWang2023,AhnHyperbolicScalarVector2020,AhnHyperbolicBlackHoles2019,
YuanGe2021,Ramirez2021,YuanGeKimJiAhn2023,AhnClassifying2021}.
More recently, the relation between pole-skipping modes and replica geometry, as well as bulk reconstruction from boundary ambiguity, has also been discussed
\cite{ChuaHartmanWeng2026,LuRanWu2026}.

The local near-horizon mechanism underlying this phenomenon is well known.
At special resonant frequencies, one step of the horizon Taylor or Frobenius recurrence degenerates, and the corresponding finite recurrence matrix loses rank.
As a result, the regular horizon solution space, which is one-dimensional at generic parameter values, may become two-dimensional at resonance.
Schematically,
\begin{equation}
	\dim \mathcal{S}_{\rm smooth}:1\longrightarrow 2.
\end{equation}
Here and below, horizon smoothness means that the radial solution extends as a $C^\infty$ function to the horizon in ingoing Eddington--Finkelstein coordinates.
In a formulation using static coordinates, where a universal ingoing factor is explicitly extracted, this is equivalent to $C^\infty$ smoothness of the residual radial function after that factor has been removed.
The local rank loss and the increase in the dimension of the smooth horizon solution space are not new results of this work.
The starting point of the present analysis is instead the following question: how is this local horizon degeneracy related to the two independently defined source and response conditions at the boundary?

\subsection{The global horizon-to-boundary problem}
\label{subsec:intro-global}

The Frobenius analysis near the horizon and the source--response decomposition at the boundary use different bases.
Near the horizon, two local Frobenius solutions arise naturally, whereas at the boundary one introduces two independent solutions
\begin{equation}
	R_1,\qquad R_2,
\end{equation}
normalized with respect to two independent asymptotic channels.
In the standard holographic setting, $R_1$ may be chosen as the source-normalized solution and $R_2$ as the response-normalized solution.
The global question in pole-skipping is therefore not merely whether two regular solutions exist at the horizon, but rather \emph{how two solutions defined independently at the boundary enter the enlarged smooth horizon solution space at resonance}.
The local near-horizon recurrence by itself does not answer this global connection problem.
For the exactly solvable JT/AdS$_2$ scalar problem, a concrete form of this global horizon-to-boundary relation was recently exhibited \cite{ChounKim2026}.
The goal of the present paper is to extend that structure to a general Fuchsian setting without relying on exact solvability.

Away from resonance, the horizon-to-boundary relation to be established in this paper already takes a simple form.
After removing a common overall normalization from the boundary coefficients,
\begin{equation}
	A=0
	\qquad\Longleftrightarrow\qquad
	R_2\in C^\infty \quad \text{at the horizon},
	\label{eq:intro-dictionary-A}
\end{equation}
and
\begin{equation}
	B=0
	\qquad\Longleftrightarrow\qquad
	R_1\in C^\infty \quad \text{at the horizon}.
	\label{eq:intro-dictionary-B}
\end{equation}
The apparent interchange is a consequence of the boundary normalization.
If the source coefficient $A$ vanishes, the source-normalized branch is absent and the response-normalized solution remains; if the response coefficient $B$ vanishes, the reverse occurs.
Thus, under suitable assumptions, a resonant point at which the source and response coefficients vanish simultaneously is related to a situation in which two solutions that are independent at the boundary both belong to the smooth resonant horizon solution space.
For this interpretation, however, the two boundary solutions must genuinely be independent, so that
\begin{equation}
	W[R_1,R_2]\neq0.
	\label{eq:intro-wronskian}
\end{equation}
The purpose of this paper is to derive this horizon-to-boundary relation for a general Fuchsian radial problem without relying on a particular exact solution or a special-function identity.

\subsection{Main results}
\label{subsec:intro-results}

The first result clarifies the local Fuchsian structure at a resonant horizon.
In ingoing Eddington--Finkelstein coordinates, define
\begin{equation}
	\nu=\frac{i\omega}{\kappa}.
\end{equation}
At the resonant frequencies
\begin{equation}
	\nu=N,
	\qquad
	N\in\mathbb{N},
\end{equation}
the two horizon indicial roots are $0$ and $N$.
In the Frobenius recurrence associated with the smaller indicial root, the coefficients with $n<N$ are determined successively, whereas at $n=N$ the recurrence coefficient $D_N$ multiplying $h_N$ vanishes.
We denote by $\Omega_N$ the left-hand side that remains as the solvability condition at this order.
For the finite $N\times N$ recurrence matrix $M_N$ used in pole-skipping analyses, we show that
\begin{equation}
	\det M_N=[(N-1)!]^2\Omega_N,
	\label{eq:intro-det}
\end{equation}
and that the logarithmic Frobenius coefficient $\ell_N$ associated with integer-separated indicial roots is
\begin{equation}
	\ell_N=-\frac{\Omega_N}{N}.
	\label{eq:intro-log}
\end{equation}
Consequently,
\begin{equation}
	\det M_N=0
	\quad\Longleftrightarrow\quad
	\Omega_N=0
	\quad\Longleftrightarrow\quad
	\ell_N=0
	\quad\Longleftrightarrow\quad
	\dim\mathcal{S}_{\rm smooth}=2 .
	\label{eq:intro-local-equivalence}
\end{equation}
Thus the singularity of the finite horizon recurrence matrix is not an independent condition; it is the matrix representation of the Frobenius solvability condition at the resonant order.
Integer resonance and logarithmic Frobenius solutions themselves belong to classical Fuchsian theory \cite{Ince1956,CoddingtonLevinson1955}, but Eq.~\eqref{eq:intro-local-equivalence} provides the local starting point for the subsequent horizon-to-boundary analysis.

The second result is the central horizon-to-boundary relation.
We write
\begin{equation}
	\lambda=(k,\Delta,\mathbf g)
\end{equation}
for the parameters other than the frequency.
If the nonresonant ingoing solution is normalized at the horizon as
\begin{equation}
	H_{\rm in}(x;\nu,\lambda)
	=
	1+\sum_{n\geq1}h_n(\nu,\lambda)x^n,
\end{equation}
then near $\nu=N$ the coefficient at the resonant order has the form
\begin{equation}
	h_N(\nu,\lambda)
	=
	\frac{\Omega_N(\nu,\lambda)}
	{N(\nu-N)}.
	\label{eq:intro-hN}
\end{equation}
Hence, at a generic integer resonance for which
\begin{equation}
	\Omega_N(N,\lambda)\neq0,
\end{equation}
the nonresonant ingoing family develops a simple pole in parameter space.
To remove this singularity, choose
\begin{equation}
	g_N(\nu,\lambda)
	=
	(\nu-N)u_N(\nu,\lambda),
	\qquad
	u_N(N,\lambda)\neq0,
\end{equation}
and define
\begin{equation}
	H_{\rm in}^{\rm reg}
	=
	g_N H_{\rm in}.
\end{equation}
We prove that at exact resonance the following identity between solutions holds:
\begin{equation}
	H_{\rm in}^{\rm reg}(x;N,\lambda)
	=
	\frac{u_N(N,\lambda)}{N}
	\Omega_N(N,\lambda)
	H_{\rm out}^{(N)}(x;\lambda),
	\label{eq:intro-central-identity}
\end{equation}
where
\begin{equation}
	H_{\rm out}^{(N)}(x;\lambda)
	=
	x^N\bigl(1+O(x)\bigr)
\end{equation}
is the resonant Frobenius solution associated with the larger indicial root $N$.
Equation~\eqref{eq:intro-central-identity} is not merely a formal power-series relation but an identity between actual analytic solutions.

Continuing this identity from the horizon to the boundary shows that the same $\Omega_N$ simultaneously controls the two boundary connection coefficients of the regularized ingoing family.
Denoting these coefficients by $\widehat A$ and $\widehat B$, and choosing the boundary basis to depend analytically on the parameters, we obtain
\begin{equation}
	\Omega_N(N,\lambda)=0
	\quad\Longleftrightarrow\quad
	\widehat A(N,\lambda)
	=
	\widehat B(N,\lambda)
	=
	0.
	\label{eq:intro-boundary-equivalence}
\end{equation}
Thus the local Frobenius solvability condition at the resonant horizon and the source--response $0/0$ structure at the boundary are governed by the same parameter-dependent quantity.
Equations~\eqref{eq:intro-central-identity} and \eqref{eq:intro-boundary-equivalence} constitute the core of the horizon-to-boundary dictionary established in this paper.

The relation $\widehat A=\widehat B=0$ does not mean that the boundary coefficients of a nonzero resonant solution both vanish.
Rather, $\widehat A$ and $\widehat B$ are the coefficient functions obtained by regularizing the nonresonant ingoing family in parameter space and then continuing it to resonance.
The resonant differential equation itself has two independent smooth solutions when $\Omega_N=0$.

The third result concerns which smooth solution at exact resonance is selected as the limit of the nonresonant ingoing solution.
At a resonant point with $\Omega_N=0$, there are two independent smooth horizon solutions, so smoothness alone does not uniquely determine the retarded solution.
With a suitable normalization, a resonant solution may be written as
\begin{equation}
	R_{\rm res}
	=
	H_0^{(N)}
	+
	\alpha_N H_{\rm out}^{(N)}.
	\label{eq:intro-res-solution}
\end{equation}

Let
\begin{equation}
	P_*=
	(\omega_N,k_*,\Delta_*,\mathbf g_*)
\end{equation}
be an exact resonant point in parameter space, with
\begin{equation}
	D_N(P_*)=0,
	\qquad
	\Omega_N(P_*)=0.
\end{equation}
Consider a $C^1$ curve approaching $P_*$,
\begin{equation}
	\gamma(s),
	\qquad
	\gamma(0)=P_*,
\end{equation}
and assume
\begin{equation}
	\left.
	\frac{d}{ds}
	D_N\!\left(\gamma(s)\right)
	\right|_{s=0}
	\neq0.
	\label{eq:intro-transverse-D}
\end{equation}
Then the coefficient selected in the resonant limit of the nonresonant ingoing solution is
\begin{equation}
	\alpha_N^{(\gamma)}
	=
	-
	\frac{
		\left.
		\dfrac{d}{ds}
		\Omega_N\left(\gamma(s)\right)
		\right|_{s=0}
	}{
		\left.
		\dfrac{d}{ds}
		D_N\left(\gamma(s)\right)
		\right|_{s=0}
	}.
	\label{eq:intro-alpha-general}
\end{equation}

In particular, if the background and surface gravity $\kappa$ are fixed and $\omega$ can be used as a local parameter along the curve, then
\begin{equation}
	\gamma(\omega)
	=
	\bigl(
	\omega,
	k(\omega),
	\Delta(\omega),
	\mathbf g
	\bigr),
	\qquad
	\gamma(\omega_N)=P_*,
\end{equation}
and
\begin{equation}
	D_N(\omega)
	=
	N\left(
	N-\frac{i\omega}{\kappa}
	\right),
	\qquad
	\frac{dD_N}{d\omega}
	=
	-\frac{iN}{\kappa}.
\end{equation}
Equation~\eqref{eq:intro-alpha-general} then reduces to
\begin{equation}
	\alpha_N^{(\gamma)}
	=
	\frac{\kappa}{iN}
	\left.
	\frac{d}{d\omega}
	\Omega_N\left(\gamma(\omega)\right)
	\right|_{\omega=\omega_N}.
	\label{eq:intro-alpha}
\end{equation}
Here the derivative is the total derivative along the parameter-space curve $\gamma$, not the partial derivative $\partial_\omega\Omega_N$.
For example, if $\Delta$ and the background are fixed while only $k=k(\omega)$ varies, then
\begin{equation}
	\left.
	\frac{d}{d\omega}
	\Omega_N\!\left(\gamma(\omega)\right)
	\right|_{\omega=\omega_N}
	=
	\left[
	\partial_\omega\Omega_N
	+
	k'(\omega_N)\partial_k\Omega_N
	\right]_{P_*}.
\end{equation}

The structure of the third result can therefore be summarized as follows.
First,
\begin{equation}
	D_N=0
\end{equation}
fixes the integer resonance, and on that resonance surface
\begin{equation}
	\Omega_N=0
\end{equation}
selects the points at which the Frobenius solvability condition is satisfied.
Which resonant solution is selected when such a point is approached is determined by the first variation of the same solvability function. Writing
$\dot\gamma_*:=\left.d\gamma/ds\right|_{s=0}$ for the tangent vector of the approach curve at $P_*$, this variation is
\begin{equation}
	\left.d\Omega_N\right|_{P_*}(\dot\gamma_*).
\end{equation}
Thus the three roles may be summarized as
\begin{equation}
	D_N=0
	\quad\text{fixes the resonance},\qquad
	\Omega_N=0
	\quad\text{fixes solvability at resonance},
\end{equation}
and
\begin{equation}
	\left.d\Omega_N\right|_{P_*}(\dot\gamma_*)
	\quad\text{controls the first-order resonant limit}.
	\label{eq:intro-value-derivative}
\end{equation}
This organizes the known direction dependence of pole-skipping in terms of the same resonant Frobenius data.

\subsection{Scope and assumptions}
\label{subsec:intro-scope}

The main results of this paper concern second-order linear scalar radial equations on backgrounds with a nonextremal horizon.
The metric coefficients and the coefficients of the radial equation are assumed to be analytic near the horizon, and the horizon is assumed to be a regular singular point in ingoing Eddington--Finkelstein coordinates.
We also assume analytic or meromorphic dependence of the coefficients on the frequency, momentum, and other external parameters.
The extremal or zero-temperature limit may have a pole-skipping structure different from that of the nonextremal case, so the present theorems are not applied directly in that regime \cite{NatsuumeOkamuraZeroTemp2021}.
General bosonic, gauge, and fermionic systems are also known to lead to multicomponent or gauge-redundant bulk equations
\cite{WangWang2022,NingWangWang2023,CeplakRamdialVegh2020,CeplakVegh2021}.
Extensions to such coupled systems and to higher-order radial equations require an analysis beyond the single scalar solvability function considered here.

For the global horizon-to-boundary statements, two independent asymptotic channels must be defined at the boundary, and the Wronskian of the corresponding solutions $R_1$ and $R_2$ must be nonzero.
The boundary itself need not be a regular singular point.
For example, if the boundary is an irregular singular point, the same reasoning applies provided that a physical ray or Stokes sector is fixed and that two independent asymptotic channels and their connection coefficients are well defined there.
If additional singularities or branch structures lie between the horizon and the boundary, the analytic continuation path and the associated monodromy data must be fixed explicitly.

The condition
\begin{equation}
	\widehat A=\widehat B=0
\end{equation}
in Eq.~\eqref{eq:intro-boundary-equivalence} means that the two boundary connection coefficients of the regularized ingoing family vanish simultaneously.
This simultaneous zero by itself does not imply that their ratio necessarily has nontrivial direction dependence near the pole-skipping point.
To obtain a stronger first-order directional ambiguity, the curves defined by $\widehat A=0$ and $\widehat B=0$ must have independent tangent directions at the pole-skipping point.

Subtle cases in which the near-horizon degeneracy alone does not determine the complete limiting behavior of boundary pole-skipping have already been noted in existing classifications \cite{AhnClassifying2021}.
In the exactly solvable JT/AdS$_2$ problem, path dependence under unrestricted approach directions also occurs \cite{ChounKim2026}.

For example, in a two-dimensional parameter slice $(\omega,k)$, a sufficient local nondegeneracy condition is
\begin{equation}
	\det
	\begin{pmatrix}
		\partial_\omega\widehat A &
		\partial_k\widehat A
		\\
		\partial_\omega\widehat B &
		\partial_k\widehat B
	\end{pmatrix}_{P_*}
	\neq 0.
	\label{eq:intro-boundary-jacobian}
\end{equation}
This excludes the case in which the two zero curves are tangent to each other at $P_*$ and ensures that they intersect with distinct tangent directions.

Under these assumptions, the results of this paper do not rely on hypergeometric solvability or on the closed-form structure of any particular metric.

\subsection{Organization of the paper}
\label{subsec:intro-organization}

The paper is organized as follows.
Section~\ref{sec:setup} formulates a general scalar radial equation on a nonextremal background in ingoing Eddington--Finkelstein coordinates and states the horizon indicial roots, the boundary asymptotic basis, and the assumptions needed for global continuation.

Section~\ref{sec:fuchsian} analyzes the recurrence at integer Frobenius resonance, defines $\Omega_N$, and proves the equivalence among the finite recurrence matrix, the logarithmic coefficient, and the smooth horizon solution space.

Section~\ref{sec:dictionary} studies the connection problem between the nonresonant horizon basis and the boundary-normalized basis, proves Eq.~\eqref{eq:intro-central-identity} for the ingoing family after removal of the parameter-dependent resonance singularity, and establishes the horizon-to-boundary dictionary.

Section~\ref{sec:resonant-limit} distinguishes the two-dimensional smooth solution space at exact resonance from the parameter-dependent resonant limit of the nonresonant ingoing solution, and analyzes which nonzero resonant solution is selected by different paths in parameter space.

Section~\ref{sec:pole-skipping-green} interprets these results in terms of pole-skipping and the holographic retarded Green function.
In particular, it distinguishes the simultaneous vanishing $\widehat A=\widehat B=0$ of the regularized boundary connection coefficients from the nonzero resonant boundary coefficient pair $(A_\gamma,B_\gamma)\neq(0,0)$ selected by a parameter-space curve $\gamma$, and discusses the first-order nondegeneracy condition under which different approach directions select different Green-function limits.

Finally, Section~\ref{sec:discussion} summarizes the roles of local Fuchsian resonance and the global horizon-to-boundary relation, discusses the scope and limitations of the results, and comments on extensions to coupled systems, extremal horizons, and more general radial problems.

\section{Scalar radial equations near a nonextremal horizon}
\label{sec:setup}

This section sets up the general scalar radial problem used in the subsequent analysis of Fuchsian resonance and the horizon-to-boundary connection.
The essential point is to distinguish two structures.
Near the horizon, we formulate a local Frobenius problem in ingoing Eddington--Finkelstein coordinates, whereas at the boundary we assume two independent asymptotic channels corresponding to the source and response.
Our subsequent goal is to analyze the connection between these two descriptions.

\subsection{Ingoing Eddington--Finkelstein geometry}
\label{subsec:setup-geometry}

Let the bulk spacetime dimension be
\begin{equation}
	D=d+2,
\end{equation}
and consider the static planar metric in ingoing Eddington--Finkelstein coordinates $(v,r,x^1,\ldots,x^d)$,
\begin{equation}
	ds^2
	=
	-A_{\rm met}(r)\,dv^2
	+
	2B_{\rm met}(r)\,dv\,dr
	+
	r^2\delta_{ij}dx^i dx^j,
	\qquad
	i,j=1,\ldots,d.
	\label{eq:setup-metric}
\end{equation}
Let the future horizon be located at
\begin{equation}
	r=r_h,
\end{equation}
and define the radial coordinate measured from the horizon by
\begin{equation}
	x:=r-r_h.
	\label{eq:setup-x}
\end{equation}

We assume that $A_{\rm met}$ and $B_{\rm met}$ are analytic near the horizon and expand them as
\begin{align}
	A_{\rm met}(r)
	&=
	A_1x+A_2x^2+A_3x^3+\cdots,
	\label{eq:setup-A-expansion}
	\\
	B_{\rm met}(r)
	&=
	B_0+B_1x+B_2x^2+\cdots.
	\label{eq:setup-B-expansion}
\end{align}
The nonextremality conditions are
\begin{equation}
	A_{\rm met}(r_h)=0,
	\qquad
	A_1=A_{\rm met}'(r_h)\neq0,
	\qquad
	B_0=B_{\rm met}(r_h)\neq0.
	\label{eq:setup-nonextremal}
\end{equation}

With this normalization, we define the surface gravity by
\begin{equation}
	\kappa
	:=
	\frac{A_1}{2B_0}.
	\label{eq:setup-kappa}
\end{equation}
The Hawking temperature is therefore
\begin{equation}
	T
	=
	\frac{\kappa}{2\pi}
	=
	\frac{A_1}{4\pi B_0}.
	\label{eq:setup-temperature}
\end{equation}
We choose the orientation of the future horizon so that $\kappa>0$.

The metric coefficients may also depend on background parameters other than the frequency and momentum.
We denote these background data collectively by $\mathbf g$.
For asymptotically AdS applications, the scalar mass may be parametrized in terms of a boundary scaling parameter $\Delta$; when convenient, we therefore abbreviate the parameters other than the frequency as
\begin{equation}
	\lambda
	=
	(k,\Delta,\mathbf g,\ldots).
	\label{eq:setup-lambda}
\end{equation}
For a more general boundary problem, the scalar mass $M$ itself may be used as a parameter in place of $\Delta$.

\subsection{Klein--Gordon radial equation in general dimension \texorpdfstring{$d$}{d}}
\label{subsec:setup-radial-equation}

Consider a probe scalar field $\Phi$ of mass $M$ on the background \eqref{eq:setup-metric}, satisfying
\begin{equation}
	(\Box-M^2)\Phi=0.
	\label{eq:setup-KG}
\end{equation}
We use the Fourier decomposition in time and the planar directions,
\begin{equation}
	\Phi(v,r,\vec x)
	=
	e^{-i\omega v+i\vec k\cdot\vec x}R(r).
	\label{eq:setup-Fourier}
\end{equation}
Here
\begin{equation}
	k^2:=\vec k\cdot\vec k,
\end{equation}
and in the complexified spectral problem this product does not involve complex conjugation.

The inverse metric components of \eqref{eq:setup-metric} are
\begin{equation}
	g^{vv}=0,
	\qquad
	g^{vr}=g^{rv}=\frac{1}{B_{\rm met}},
	\qquad
	g^{rr}=\frac{A_{\rm met}}{B_{\rm met}^2},
	\qquad
	g^{ij}=\frac{\delta^{ij}}{r^2},
	\label{eq:setup-inverse-metric}
\end{equation}
and, for the chosen EF orientation,
\begin{equation}
	\sqrt{-g}
	=
	B_{\rm met}(r)r^d.
	\label{eq:setup-determinant}
\end{equation}

Substituting these expressions into Eq.~\eqref{eq:setup-KG} gives the radial equation
\begin{equation}
	\frac{1}{r^d}
	\frac{d}{dr}
	\left(
	\frac{A_{\rm met}(r)r^d}{B_{\rm met}(r)}
	R'(r)
	\right)
	-
	2i\omega R'(r)
	-
	\frac{i\omega d}{r}R(r)
	-
	B_{\rm met}(r)
	\left(
	M^2+\frac{k^2}{r^2}
	\right)R(r)
	=
	0.
	\label{eq:setup-radial}
\end{equation}

For convenience, define
\begin{equation}
	F(r)
	:=
	\frac{A_{\rm met}(r)}{B_{\rm met}(r)}.
	\label{eq:setup-F}
\end{equation}
Then Eq.~\eqref{eq:setup-radial} can be written as
\begin{equation}
	F R''
	+
	\left(
	F'
	+
	\frac{d}{r}F
	-
	2i\omega
	\right)R'
	-
	\left[
	\frac{i\omega d}{r}
	+
	B_{\rm met}
	\left(
	M^2+\frac{k^2}{r^2}
	\right)
	\right]R
	=
	0.
	\label{eq:setup-radial-F}
\end{equation}

Equations~\eqref{eq:setup-A-expansion} and \eqref{eq:setup-B-expansion} imply
\begin{equation}
	F(r)
	=
	\frac{A_1}{B_0}x+O(x^2)
	=
	2\kappa x+O(x^2).
	\label{eq:setup-F-horizon}
\end{equation}
Hence the leading part of Eq.~\eqref{eq:setup-radial-F} near the horizon is
\begin{equation}
	\frac{A_1}{B_0}
	\left(
	xR''+R'
	\right)
	-
	2i\omega R'
	+
	O(1)R'
	+
	O(1)R
	=
	0.
	\label{eq:setup-leading-equation}
\end{equation}

\subsection{Regular singular horizon and indicial roots}
\label{subsec:setup-indicial}

Write Eq.~\eqref{eq:setup-radial-F} in standard form,
\begin{equation}
	R''
	+
	P(x)R'
	+
	Q(x)R
	=
	0.
	\label{eq:setup-standard-form}
\end{equation}
Using Eq.~\eqref{eq:setup-F-horizon}, we find
\begin{equation}
	P(x)
	=
	\frac{1-\nu}{x}+O(1),
	\qquad
	\nu
	:=
	\frac{i\omega}{\kappa},
	\label{eq:setup-P-leading}
\end{equation}
and
\begin{equation}
	Q(x)
	=
	O\left(\frac{1}{x}\right).
	\label{eq:setup-Q-leading}
\end{equation}
Thus
\begin{equation}
	xP(x),
	\qquad
	x^2Q(x)
\end{equation}
are analytic near $x=0$.
Therefore $x=0$, or equivalently $r=r_h$, is a regular singular point of Eq.~\eqref{eq:setup-standard-form}
\cite{Ince1956,CoddingtonLevinson1955}.

\begin{proposition}[Indicial roots at a nonextremal horizon]
	\label{prop:setup-indicial-roots}
	Under the assumptions \eqref{eq:setup-nonextremal}, the horizon indicial equation of the scalar radial equation \eqref{eq:setup-radial} is
	\begin{equation}
		I_h(\rho)
		=
		\rho(\rho-\nu)
		=
		0,
		\qquad
		\nu=\frac{i\omega}{\kappa}.
		\label{eq:setup-indicial-equation}
	\end{equation}
	Hence the two indicial roots are
	\begin{equation}
		\rho_{\rm in}=0,
		\qquad
		\rho_{\rm out}=\nu
		=
		\frac{i\omega}{\kappa}
		=
		\frac{i\omega}{2\pi T}.
		\label{eq:setup-indicial-roots}
	\end{equation}
\end{proposition}

\begin{proof}
	Substituting
	\begin{equation}
		R(x)\sim x^\rho
	\end{equation}
	into Eq.~\eqref{eq:setup-leading-equation}, the most singular terms give
	\begin{equation}
		\rho
		\left(
		\frac{A_1}{B_0}\rho-2i\omega
		\right)
		=
		0.
	\end{equation}
	Using $A_1/B_0=2\kappa$ from Eq.~\eqref{eq:setup-kappa}, this becomes
	\begin{equation}
		\rho
		\left(
		\rho-\frac{i\omega}{\kappa}
		\right)
		=
		0,
	\end{equation}
	which gives Eq.~\eqref{eq:setup-indicial-roots}.
\end{proof}

An important point is that the difference between the two indicial roots,
\begin{equation}
	\rho_{\rm out}-\rho_{\rm in}
	=
	\nu,
\end{equation}
is independent of the spatial dimension $d$.
The quantities $d$, $M$, $k$, and the higher Taylor coefficients of the background enter the subleading recurrence coefficients and the actual pole-skipping condition, but not the indicial-root difference responsible for the horizon resonance.

Integer resonance occurs when
\begin{equation}
	\nu=N,
	\qquad
	N\in\mathbb N=\{1,2,\ldots\}.
	\label{eq:setup-resonance-N}
\end{equation}
The corresponding frequencies are
\begin{equation}
	\omega_N
	=
	-iN\kappa
	=
	-i2\pi TN,
	\qquad
	N\in\mathbb N,
	\label{eq:setup-resonant-frequency}
\end{equation}
and at exact resonance the indicial roots are
\begin{equation}
	\rho_{\rm in}=0,
	\qquad
	\rho_{\rm out}=N.
	\label{eq:setup-resonant-roots}
\end{equation}

\subsection{Horizon branches and smoothness}
\label{subsec:setup-horizon-branches}

First consider the nonresonant case
\begin{equation}
	\nu\notin\mathbb Z.
\end{equation}
Normalizing the leading coefficient of each local Frobenius solution to unity, we write
\begin{align}
	H_{\rm in}(x;\nu,\lambda)
	&=
	1+h_1x+h_2x^2+\cdots,
	\label{eq:setup-Hin}
	\\
	H_{\rm out}(x;\nu,\lambda)
	&=
	x^\nu
	\left(
	1+b_1x+b_2x^2+\cdots
	\right).
	\label{eq:setup-Hout}
\end{align}
A branch of $x^\nu$ is fixed consistently with the chosen radial continuation path.

The exponent-zero branch in Eq.~\eqref{eq:setup-Hin} extends regularly through the future horizon in ingoing Eddington--Finkelstein coordinates.
Thus no separate universal ingoing radial factor needs to be removed in these coordinates.
Equivalently, in the EF formulation the universal ingoing factor is effectively
\begin{equation}
	\chi_{\rm in}^{\rm EF}=1.
\end{equation}

This may be compared with a formulation in static coordinates.
If the static time $t$ and the EF time $v$ are related by
\begin{equation}
	v=t+r_*,
\end{equation}
and the static-coordinate ingoing factor is
\begin{equation}
	\chi_{\rm in}^{(t)}
	=
	e^{-i\omega r_*},
\end{equation}
then
\begin{equation}
	R_t
	=
	\chi_{\rm in}^{(t)}R_{\rm EF}.
	\label{eq:setup-static-EF}
\end{equation}
Hence the condition that ``the residual radial function is smooth after the ingoing factor has been removed'' in static coordinates is equivalent to smoothness of the radial function itself in EF coordinates.
This is one reason for using EF coordinates throughout the present analysis
\cite{SonStarinets2002,HerzogSon2003,SkenderisVanRees2008,VanRees2009,KovtunStarinets2005}.

\begin{definition}[Horizon smoothness]
	\label{def:setup-horizon-smoothness}
	A radial solution $R$ is said to be smooth at the horizon if, as a function of the ingoing Eddington--Finkelstein radial coordinate $x=r-r_h$,
	\begin{equation}
		R\in C^\infty
		\qquad
		\text{at }x=0.
		\label{eq:setup-smoothness}
	\end{equation}
\end{definition}

In the nonresonant case, $H_{\rm in}$ is analytic and hence horizon-smooth.
By contrast, for generic
\begin{equation}
	\nu\notin\mathbb Z,
\end{equation}
the factor $x^\nu$ is not $C^\infty$ at $x=0$, so $H_{\rm out}$ is not horizon-smooth.

At exact integer resonance,
\begin{equation}
	\nu=N,
\end{equation}
the situation changes.
There always exists a local Frobenius solution associated with the larger root of the form
\begin{equation}
	H_{\rm out}^{(N)}(x;\lambda)
	=
	x^N
	\left(
	1+b_1x+b_2x^2+\cdots
	\right),
	\label{eq:setup-Hout-resonant}
\end{equation}
and this solution is analytic.
The smaller-root branch, however, may contain a logarithmic term because the indicial roots differ by an integer.
Thus the existence of two independent smooth horizon solutions at resonance is not determined by the indicial roots alone.
We analyze this question in Section~\ref{sec:fuchsian} using the resonant Frobenius recurrence.

When no confusion can arise, we use the same symbols $H_{\rm in}$ and $H_{\rm out}$ both for the local solutions defined near the horizon and for their global continuations along the chosen radial path.
Whenever necessary, the local germ and its global continuation will be distinguished explicitly in the text.

\subsection{Boundary asymptotic basis and connection coefficients}
\label{subsec:setup-boundary-basis}

Let the boundary endpoint be denoted by $\mathcal I$.
The horizon-to-boundary theorem developed below does not require $\mathcal I$ itself to be a regular singular point.
What is needed is the existence of two independent asymptotic solutions in a chosen asymptotic region, together with a unique expansion of any solution continued from the horizon in terms of those two solutions.

\begin{definition}[Admissible boundary basis]
	\label{def:setup-boundary-admissibility}
	Fix an asymptotic ray or Stokes sector $\mathcal S$ at the boundary endpoint $\mathcal I$.
	Two asymptotic solutions
	\begin{equation}
		\mathcal B_{\rm src}^{(\mathcal S)},
		\qquad
		\mathcal B_{\rm rsp}^{(\mathcal S)}
		\label{eq:setup-boundary-channels}
	\end{equation}
	will be called an admissible boundary basis if the following conditions hold.

	\begin{enumerate}
		\item
		The functions $\mathcal B_{\rm src}^{(\mathcal S)}$ and $\mathcal B_{\rm rsp}^{(\mathcal S)}$ are linearly independent in the chosen sector $\mathcal S$.

		\item
		At an ordinary point $r_0$ where both solutions are defined,
		\begin{equation}
			W\!\left[
			\mathcal B_{\rm src}^{(\mathcal S)},
			\mathcal B_{\rm rsp}^{(\mathcal S)}
			\right](r_0)\neq0,
			\label{eq:setup-boundary-Wronskian}
		\end{equation}
		where
		\begin{equation}
			W[f,g]:=fg'-f'g.
		\end{equation}

		\item
		A relevant solution continued from the horizon along a fixed continuation path has a unique asymptotic expansion in the boundary sector $\mathcal S$,
		\begin{equation}
			R
			\sim
			A\,\mathcal B_{\rm src}^{(\mathcal S)}
			+
			B\,\mathcal B_{\rm rsp}^{(\mathcal S)},
			\label{eq:setup-boundary-decomposition}
		\end{equation}
		where $A$ and $B$ are well-defined parameter-dependent connection coefficients.
	\end{enumerate}
\end{definition}

Let $R_1$ and $R_2$ denote the exact solutions corresponding to the asymptotic solutions in Eq.~\eqref{eq:setup-boundary-channels}, normalized so that
\begin{align}
	R_1
	&\sim
	\mathcal B_{\rm src}^{(\mathcal S)},
	\label{eq:setup-R1-boundary}
	\\
	R_2
	&\sim
	\mathcal B_{\rm rsp}^{(\mathcal S)}.
	\label{eq:setup-R2-boundary}
\end{align}
If these two solutions are independent in the chosen continuation domain, every global solution can be written exactly as
\begin{equation}
	R
	=
	A R_1+B R_2.
	\label{eq:setup-global-decomposition}
\end{equation}
The coefficients $A$ and $B$ in Eq.~\eqref{eq:setup-boundary-decomposition} are precisely the connection coefficients in Eq.~\eqref{eq:setup-global-decomposition}.

In a standard asymptotically AdS problem, $\mathcal B_{\rm src}$ and $\mathcal B_{\rm rsp}$ may be chosen as the source and response asymptotic behaviors, and $R_1$ and $R_2$ then become the source-normalized and response-normalized solutions, respectively.
The general theorem below does not, however, require these channels to be power-law Frobenius solutions.

If the boundary is a regular singular point, one may choose Frobenius channels of the form
\begin{equation}
	R_\pm(z)
	=
	z^{\rho_\pm}S_\pm(z).
\end{equation}
In that case, one may impose $C^\infty$ regularity on the residual functions $S_\pm$ after the leading powers have been removed.
If the boundary indicial roots differ by an integer, this regularity condition requires the absence of the logarithmic term.
Such regular-singular boundaries form an important subclass of Definition~\ref{def:setup-boundary-admissibility}, but they are not required.

Conversely, if the boundary is an irregular singular point, Definition~\ref{def:setup-boundary-admissibility} still applies after fixing a Stokes sector $\mathcal S$, provided that two independent asymptotic channels and their connection coefficients are well defined there.
Since the connection coefficients may differ between Stokes sectors through Stokes multipliers, the choice of sector $\mathcal S$ is regarded as part of the global data in an irregular problem.

\subsection{Global continuation and standing assumptions}
\label{subsec:setup-global-assumptions}

The theorems below connect the local Frobenius analysis at the horizon to the asymptotic basis data at the boundary.
We use the following conditions as standing assumptions.

\begin{enumerate}
	\item[(G1)]
	The functions $A_{\rm met}$ and $B_{\rm met}$ are analytic near the horizon and satisfy the nonextremality conditions \eqref{eq:setup-nonextremal}.

	\item[(G2)]
	The radial equation is a second-order linear equation of the form \eqref{eq:setup-radial}, and its coefficients depend analytically or meromorphically on the parameters in the parameter neighborhood of interest.

	\item[(G3)]
	A continuation path $\Gamma$ from the horizon to the boundary sector $\mathcal S$ is fixed.
	In the simplest case there are no additional singular points between the horizon and the boundary.
	If additional singularities, branch cuts, or nontrivial monodromy are present, $\Gamma$ and the corresponding monodromy data are fixed explicitly as part of the global problem.

	\item[(G4)]
	There exist two independent channels at the boundary satisfying Definition~\ref{def:setup-boundary-admissibility}, and the corresponding solutions $R_1$ and $R_2$ can be continued to the horizon.
	In particular, we require
	\begin{equation}
		W[R_1,R_2](r_0)\neq0.
		\label{eq:setup-W-R1R2}
	\end{equation}

	\item[(G5)]
	In a neighborhood of the resonance of interest, the boundary basis and the horizon-to-boundary connection map can be chosen to depend analytically on the parameters.
	If an initially chosen boundary basis itself develops a parameter pole, a logarithmic degeneration, or a Stokes degeneration, one first chooses a boundary basis that depends analytically on the parameters and then defines the connection coefficients with respect to that basis.

	\item[(G6)]
	In a sufficiently small punctured parameter neighborhood of the integer resonance of interest, namely a neighborhood with the resonant point removed, the nonresonant ingoing solution is well defined up to normalization and can be followed as a function of the parameters.
	The parameter singularity that actually develops as resonance is approached is derived directly from the Frobenius recurrence in Section~\ref{sec:dictionary}.
\end{enumerate}

Connection coefficients depend on the overall normalization of the solution.
For example, under
\begin{equation}
	H_{\rm in}
	\longrightarrow
	f(\nu,\lambda)H_{\rm in},
	\qquad
	f\neq0,
\end{equation}
one has
\begin{equation}
	(A,B)
	\longrightarrow
	f(\nu,\lambda)(A,B).
\end{equation}
Thus rescaling by a nonvanishing analytic factor does not change the source--response ratio or the zero condition of either coefficient.
By contrast, a parameter-dependent factor with a zero or pole at resonance can change the zero structure itself and must not be removed arbitrarily.
The simple parameter pole that appears at resonance, together with the analytic factor needed to remove it, is derived directly from the recurrence in Section~\ref{sec:dictionary}.

To avoid notational ambiguity, we adopt the following convention.
The ordinary connection coefficients obtained by continuing the nonresonant ingoing solution to the boundary are denoted by
\begin{equation}
	A,\qquad B.
\end{equation}
The connection coefficient functions obtained after applying the parameter regularization required for the ingoing family near resonance are denoted by
\begin{equation}
	\widehat A,\qquad \widehat B.
\end{equation}

The relation to be proved below,
\begin{equation}
	\widehat A(P_*)=\widehat B(P_*)=0,
\end{equation}
does not mean that a nonzero resonant solution has both boundary coefficients equal to zero.
Indeed, if $R_1$ and $R_2$ are independent, any nonzero resonant solution must be written as
\begin{equation}
	R_{\rm res}
	=
	A_*R_1+B_*R_2,
	\qquad
	(A_*,B_*)\neq(0,0).
\end{equation}
The simultaneous zero above instead states, in terms of boundary coefficients, that the parameter-dependent regularized family constructed from the nonresonant ingoing solution has the zero solution as its value at $P_*$.

\section{Fuchsian structure at a resonant horizon}
\label{sec:fuchsian}

In this section, we recast the local radial equation near the nonextremal horizon obtained in Section~\ref{sec:setup} in an abstract Fuchsian form and analyze the Frobenius recurrence at integer resonance.
The analysis is entirely local to the horizon and does not use boundary asymptotics or horizon-to-boundary continuation.

Our main purpose is to show that a single solvability condition at the resonant order is equivalent to the singularity of the finite recurrence matrix, the absence of the logarithmic Frobenius term, and the increase in the dimension of the smooth horizon solution space.

\subsection{Standard Fuchs form and Frobenius recurrence}
\label{subsec:fuchs-standard}

As shown in Section~\ref{subsec:setup-indicial}, the nonextremal horizon $x=r-r_h=0$ is a regular singular point of the radial equation.
After multiplying or dividing the equation by a nonvanishing analytic function near the horizon, it can be written in the standard Fuchs form
\begin{equation}
	LR
	:=
	x^2R''+x p(x)R'+q(x)R
	=
	0.
	\label{eq:fuchs-operator}
\end{equation}
Here
\begin{equation}
	p(x)
	=
	\sum_{j=0}^{\infty}p_jx^j,
	\qquad
	q(x)
	=
	\sum_{j=0}^{\infty}q_jx^j
	\label{eq:fuchs-pq}
\end{equation}
are analytic near $x=0$.
When necessary, $p_j$ and $q_j$ depend on $(\omega,\lambda)$, but we suppress this parameter dependence throughout the present section for notational simplicity.

The indicial polynomial of Eq.~\eqref{eq:fuchs-operator} is
\begin{equation}
	I(\rho)
	=
	\rho(\rho-1)+p_0\rho+q_0.
	\label{eq:fuchs-indicial}
\end{equation}
At integer resonance,
\begin{equation}
	\nu=N,
	\qquad
	N\in\mathbb N,
	\label{eq:fuchs-resonance}
\end{equation}
the two indicial roots found in Section~\ref{subsec:setup-indicial} are $0$ and $N$, so
\begin{equation}
	I(\rho)
	=
	\rho(\rho-N).
	\label{eq:fuchs-indicial-resonant}
\end{equation}
Consequently,
\begin{equation}
	p_0=1-N,
	\qquad
	q_0=0.
	\label{eq:fuchs-p0q0}
\end{equation}

Consider the formal Frobenius series associated with the smaller indicial root $\rho=0$,
\begin{equation}
	R(x)
	=
	\sum_{n=0}^{\infty}a_nx^n,
	\qquad
	a_0\neq0.
	\label{eq:fuchs-small-series}
\end{equation}
Substitution into Eq.~\eqref{eq:fuchs-operator} gives, for $n\geq1$,
\begin{equation}
	n(n-N)a_n
	+
	\sum_{j=1}^{n}
	\left[
	(n-j)p_j+q_j
	\right]
	a_{n-j}
	=
	0.
	\label{eq:fuchs-recurrence}
\end{equation}

For $1\leq n\leq N-1$,
\begin{equation}
	n(n-N)\neq0,
	\label{eq:fuchs-nonzero-pivots}
\end{equation}
so once $a_0$ is specified the coefficients
\begin{equation}
	a_1,a_2,\ldots,a_{N-1}
\end{equation}
are determined successively and uniquely.
At precisely $n=N$, however,
\begin{equation}
	N(N-N)=0,
	\label{eq:fuchs-resonant-pivot}
\end{equation}
so the recurrence coefficient $D_N$ multiplying $a_N$ vanishes.
This is the basic local degeneracy associated with integer Frobenius resonance
\cite{Ince1956,CoddingtonLevinson1955}.

\subsection{Solvability condition at the resonant order}
\label{subsec:fuchs-solvability}

Since the overall normalization is arbitrary, we set
\begin{equation}
	a_0=1.
	\label{eq:fuchs-a0-normalization}
\end{equation}
Then the recurrence equations \eqref{eq:fuchs-recurrence} with $n=1,\ldots,N-1$ determine $a_1,\ldots,a_{N-1}$ uniquely.

Substituting these values into the recurrence at $n=N$, the coefficient $a_N$ has already dropped out, and the only remaining condition is
\begin{equation}
	\sum_{j=1}^{N}
	\left[
	(N-j)p_j+q_j
	\right]
	a_{N-j}
	=
	0.
	\label{eq:fuchs-resonant-equation}
\end{equation}

\begin{definition}[Function associated with the solvability condition at the resonant order]
	\label{def:fuchs-Omega}
	Normalize the smaller-root solution by $a_0=1$ and determine $a_1,\ldots,a_{N-1}$ from the recurrences at $n=1,\ldots,N-1$.
	We define the left-hand side remaining in the $n=N$ equation by
	\begin{equation}
		\Omega_N
		:=
		\sum_{j=1}^{N}
		\left[
		(N-j)p_j+q_j
		\right]
		a_{N-j}.
		\label{eq:fuchs-Omega}
	\end{equation}
	Thus the vanishing of $\Omega_N$ is precisely the solvability condition at the resonant order.
	For brevity, we refer to $\Omega_N$ below as the solvability function.
	The condition for the smaller-root power-series solution to pass through the resonant order $n=N$ without a logarithm is therefore
	\begin{equation}
		\Omega_N=0.
		\label{eq:fuchs-Omega-zero}
	\end{equation}
\end{definition}

It is important that $\Omega_N$ is not an independently introduced dynamical quantity.
Its meaning is exactly
\begin{equation}
	\boxed{
		\begin{aligned}
			\Omega_N={}&
			\text{the quantity remaining in the $n=N$ recurrence}\\
			&\text{after solving the first $N-1$ recurrences}
		\end{aligned}
	}
	\label{eq:fuchs-Omega-meaning}
\end{equation}
.

Thus
\begin{equation}
	D_N=0
\end{equation}
identifies the integer resonance itself, whereas
\begin{equation}
	\Omega_N=0
\end{equation}
determines whether, once resonance has occurred, the smaller-root Frobenius series can be continued without a logarithmic correction.
The two conditions therefore play distinct roles.

\subsection{Finite recurrence matrix}
\label{subsec:fuchs-matrix}

We now write the first $N$ recurrence equations as a finite-dimensional linear system.
Define
\begin{equation}
	D_n
	:=
	n(n-N),
	\label{eq:fuchs-Dn}
\end{equation}
and, for $0\leq m\leq n-1$,
\begin{equation}
	C_{n,m}
	:=
	mp_{n-m}+q_{n-m}.
	\label{eq:fuchs-Cnm}
\end{equation}
Then Eq.~\eqref{eq:fuchs-recurrence} becomes
\begin{equation}
	C_{n,0}a_0
	+
	C_{n,1}a_1
	+\cdots+
	C_{n,n-1}a_{n-1}
	+
	D_na_n
	=
	0.
	\label{eq:fuchs-recurrence-C}
\end{equation}

For example, the first three equations are
\begin{align}
	n=1:\quad&
	q_1a_0+D_1a_1=0,
	\label{eq:fuchs-first-rec}
	\\
	n=2:\quad&
	q_2a_0+(p_1+q_1)a_1+D_2a_2=0,
	\label{eq:fuchs-second-rec}
	\\
	n=3:\quad&
	q_3a_0+(p_2+q_2)a_1
	+(2p_1+q_1)a_2+D_3a_3=0.
	\label{eq:fuchs-third-rec}
\end{align}

At exact resonance $D_N=0$.
Collecting the unknowns appearing in the first $N$ equations into
\begin{equation}
	\mathbf u
	=
	\begin{pmatrix}
		a_0\\
		a_1\\
		\vdots\\
		a_{N-1}
	\end{pmatrix},
	\label{eq:fuchs-u}
\end{equation}
we obtain
\begin{equation}
	M_N\mathbf u=0,
	\label{eq:fuchs-matrix-system}
\end{equation}
where
\begin{equation}
	M_N
	=
	\begin{pmatrix}
		C_{1,0} & D_1 & 0 & \cdots & 0
		\\
		C_{2,0} & C_{2,1} & D_2 & \cdots & 0
		\\
		C_{3,0} & C_{3,1} & C_{3,2} & \ddots & \vdots
		\\
		\vdots & \vdots & \vdots & \ddots & D_{N-1}
		\\
		C_{N,0} & C_{N,1} & C_{N,2} & \cdots & C_{N,N-1}
	\end{pmatrix}.
	\label{eq:fuchs-MN}
\end{equation}

Thus $M_N$ is not a new matrix introduced independently; rather,
\begin{equation}
	\boxed{
		M_N
		=
		\text{the coefficient matrix of the first $N$ Frobenius recurrence equations}
	}
	\label{eq:fuchs-MN-meaning}
\end{equation}
.
The reason $D_N$ does not appear in the final row is precisely that $D_N=0$ at resonance.

\subsection{Equivalence of the recurrence-matrix condition and the solvability condition}
\label{subsec:fuchs-det-omega}

We now prove directly the equivalence between singularity of $M_N$ and the condition $\Omega_N=0$.

Introduce
\begin{equation}
	\mathbf a
	=
	\begin{pmatrix}
		a_1\\
		\vdots\\
		a_{N-1}
	\end{pmatrix}
\end{equation}
and write $M_N$ in block form as
\begin{equation}
	M_N
	=
	\begin{pmatrix}
		\mathbf c & P
		\\
		s & \mathbf r^{\,T}
	\end{pmatrix}.
	\label{eq:fuchs-block-M}
\end{equation}
Here
\begin{equation}
	\mathbf c
	=
	\begin{pmatrix}
		C_{1,0}\\
		\vdots\\
		C_{N-1,0}
	\end{pmatrix}
	=
	\begin{pmatrix}
		q_1\\
		\vdots\\
		q_{N-1}
	\end{pmatrix},
	\qquad
	s=C_{N,0}=q_N,
	\label{eq:fuchs-cs}
\end{equation}
and
\begin{equation}
	\mathbf r^{\,T}
	=
	\begin{pmatrix}
		C_{N,1}&C_{N,2}&\cdots&C_{N,N-1}
	\end{pmatrix}.
	\label{eq:fuchs-r}
\end{equation}

The matrix $P$ is the lower-triangular matrix multiplying $a_1,\ldots,a_{N-1}$ in the first $N-1$ equations, and its diagonal entries are
\begin{equation}
	D_1,D_2,\ldots,D_{N-1}.
\end{equation}
Hence
\begin{align}
	\det P
	&=
	\prod_{n=1}^{N-1}n(n-N)
	\nonumber\\
	&=
	(-1)^{N-1}[(N-1)!]^2
	\neq0.
	\label{eq:fuchs-detP}
\end{align}
Therefore $P^{-1}$ exists.

For $a_0=1$, the first $N-1$ recurrence equations are
\begin{equation}
	\mathbf c+P\mathbf a=0,
\end{equation}
so
\begin{equation}
	\mathbf a
	=
	-P^{-1}\mathbf c.
	\label{eq:fuchs-a-solved}
\end{equation}

The final resonant equation is
\begin{equation}
	s+\mathbf r^{\,T}\mathbf a=0.
\end{equation}
Substituting Eq.~\eqref{eq:fuchs-a-solved} gives
\begin{equation}
	s-\mathbf r^{\,T}P^{-1}\mathbf c=0.
	\label{eq:fuchs-schur-condition}
\end{equation}
By Definition~\ref{def:fuchs-Omega},
\begin{equation}
	\boxed{
		\Omega_N
		=
		s-\mathbf r^{\,T}P^{-1}\mathbf c
	}
	\label{eq:fuchs-Omega-block}
\end{equation}
.

\begin{proposition}[Recurrence matrix and solvability condition]
	\label{prop:fuchs-det-Omega}
	In the standard Fuchs normalization \eqref{eq:fuchs-operator} with $a_0=1$,
	\begin{equation}
		\boxed{
			\det M_N
			=
			[(N-1)!]^2\Omega_N
		}
		\label{eq:fuchs-det-exact}
	\end{equation}
	and therefore
	\begin{equation}
		\det M_N=0
		\qquad\Longleftrightarrow\qquad
		\Omega_N=0.
		\label{eq:fuchs-det-Omega-equivalence}
	\end{equation}
\end{proposition}

\begin{proof}
	Moving the first column in Eq.~\eqref{eq:fuchs-block-M} to the last position requires $N-1$ column interchanges, so
	\begin{equation}
		\det M_N
		=
		(-1)^{N-1}
		\det
		\begin{pmatrix}
			P & \mathbf c
			\\
			\mathbf r^{\,T} & s
		\end{pmatrix}.
		\label{eq:fuchs-column-swap}
	\end{equation}

	Since $P$ is invertible, subtracting from the final row the linear combination of the first $N-1$ rows determined by $\mathbf r^{\,T}P^{-1}$ leaves the determinant unchanged and gives
	\begin{equation}
		\begin{pmatrix}
			P & \mathbf c
			\\
			\mathbf r^{\,T} & s
		\end{pmatrix}
		\longrightarrow
		\begin{pmatrix}
			P & \mathbf c
			\\
			0 & s-\mathbf r^{\,T}P^{-1}\mathbf c
		\end{pmatrix}.
		\label{eq:fuchs-row-elimination}
	\end{equation}
	Thus
	\begin{equation}
		\det
		\begin{pmatrix}
			P & \mathbf c
			\\
			\mathbf r^{\,T} & s
		\end{pmatrix}
		=
		(\det P)\Omega_N.
		\label{eq:fuchs-schur-det}
	\end{equation}
	Using Eq.~\eqref{eq:fuchs-detP}, we obtain
	\begin{align}
		\det M_N
		&=
		(-1)^{N-1}
		(-1)^{N-1}
		[(N-1)!]^2
		\Omega_N
		\nonumber\\
		&=
		[(N-1)!]^2\Omega_N.
	\end{align}
\end{proof}

The expression in Eq.~\eqref{eq:fuchs-Omega-block} may be viewed as the Schur complement of $P$ in linear algebra.
What matters here, however, is not the terminology but the fact that $\Omega_N$ is precisely
\begin{equation}
	\begin{aligned}
	&\text{the quantity remaining in the final equation}\\
	&\text{after eliminating the first $N-1$ recurrences}.
	\end{aligned}
\end{equation}

\subsection{Logarithmic Frobenius solution}
\label{subsec:fuchs-log}

We now relate $\Omega_N$ to the logarithmic term that appears at integer resonance.

Normalize the Frobenius solution associated with the larger indicial root $N$ as
\begin{equation}
	F_N(x)
	=
	x^N
	\left(
	1+b_1x+b_2x^2+\cdots
	\right).
	\label{eq:fuchs-FN}
\end{equation}
Writing an order above $N$ as $N+m$, the corresponding recurrence coefficient is
\begin{equation}
	(N+m)\bigl((N+m)-N\bigr)
	=
	m(N+m)
	\neq0,
	\qquad
	m\geq1.
	\label{eq:fuchs-larger-pivots}
\end{equation}
Hence $b_1,b_2,\ldots$ are determined successively and uniquely.

For integer-separated indicial roots, the smaller-root solution generally has the form
\begin{equation}
	F_0(x)
	=
	H(x)
	+
	\ell_N F_N(x)\log x,
	\label{eq:fuchs-F0-log}
\end{equation}
as in Refs.~\cite{Ince1956,CoddingtonLevinson1955}.
Here $H(x)$ is the analytic power-series part and $\ell_N$ is the logarithmic Frobenius coefficient.

Write the smaller-root power-series part as
\begin{equation}
	H(x)
	=
	a_0+a_1x+\cdots.
\end{equation}
After the preceding recurrence relations have fixed $a_0,\ldots,a_{N-1}$, one has
\begin{equation}
	LH
	=
	\Omega_N x^N
	+
	O(x^{N+1}).
	\label{eq:fuchs-LH}
\end{equation}

Using Eq.~\eqref{eq:fuchs-operator},
\begin{equation}
	L[F_N\log x]
	=
	(\log x)L[F_N]
	+
	2xF_N'
	+
	\bigl(p(x)-1\bigr)F_N.
	\label{eq:fuchs-LFNlog}
\end{equation}
Since $F_N$ is an exact solution,
\begin{equation}
	L[F_N]=0.
\end{equation}
Moreover,
\begin{equation}
	F_N=x^N+O(x^{N+1}),
	\qquad
	F_N'=Nx^{N-1}+O(x^N),
\end{equation}
and $p_0=1-N$, so
\begin{align}
	L[F_N\log x]
	&=
	2Nx^N
	+
	(p_0-1)x^N
	+
	O(x^{N+1})
	\nonumber\\
	&=
	Nx^N+O(x^{N+1}).
	\label{eq:fuchs-LFNlog-leading}
\end{align}

Applying $L$ to Eq.~\eqref{eq:fuchs-F0-log} gives
\begin{align}
	0
	=
	LF_0
	&=
	LH
	+
	\ell_N L[F_N\log x]
	\nonumber\\
	&=
	\left(
	\Omega_N+N\ell_N
	\right)x^N
	+
	O(x^{N+1}).
\end{align}
Therefore
\begin{equation}
	\boxed{
		\ell_N
		=
		-\frac{\Omega_N}{N}
	}
	\label{eq:fuchs-log-Omega}
\end{equation}
.

Combining this with Proposition~\ref{prop:fuchs-det-Omega},
\begin{equation}
	\ell_N
	=
	-\frac{\det M_N}
	{N[(N-1)!]^2}.
	\label{eq:fuchs-log-det}
\end{equation}
Hence
\begin{equation}
	\boxed{
		\det M_N=0
		\quad\Longleftrightarrow\quad
		\Omega_N=0
		\quad\Longleftrightarrow\quad
		\ell_N=0
	}
	\label{eq:fuchs-three-equivalence}
\end{equation}
.

\subsection{Smooth horizon solution space}
\label{subsec:fuchs-smooth-space}

We now relate the logarithmic coefficient directly to horizon smoothness.
At exact resonance, choose the two local solutions as
\begin{equation}
	F_N
	=
	x^N(1+O(x)),
	\qquad
	F_0
	=
	H+\ell_NF_N\log x.
	\label{eq:fuchs-two-local-solutions}
\end{equation}
The solution $F_N$ is always analytic.

First suppose
\begin{equation}
	\ell_N\neq0.
\end{equation}
A general local solution is
\begin{equation}
	R
	=
	C_0F_0+C_NF_N,
\end{equation}
so
\begin{equation}
	R
	=
	C_0H+C_NF_N
	+
	C_0\ell_NF_N\log x.
	\label{eq:fuchs-general-log-solution}
\end{equation}
Since $F_N=x^N(1+O(x))$, the logarithmic part contains
\begin{equation}
	x^N\log x,
\end{equation}
which is not $C^\infty$ at $x=0$.
Therefore horizon smoothness in the sense of Definition~\ref{def:setup-horizon-smoothness} requires
\begin{equation}
	C_0\ell_N=0.
\end{equation}
If $\ell_N\neq0$, then
\begin{equation}
	C_0=0,
\end{equation}
and the only smooth direction is
\begin{equation}
	R=C_NF_N.
\end{equation}
Thus
\begin{equation}
	\ell_N\neq0
	\qquad\Longrightarrow\qquad
	\dim\mathcal S_{\rm smooth}=1.
	\label{eq:fuchs-smooth-dim-one}
\end{equation}

Conversely, if
\begin{equation}
	\ell_N=0,
\end{equation}
then
\begin{equation}
	F_0=H,
	\qquad
	F_N=x^N(1+O(x))
\end{equation}
are both analytic and are linearly independent because their indicial behaviors differ.
Hence every linear combination
\begin{equation}
	R=C_0F_0+C_NF_N
\end{equation}
is smooth at the horizon, and
\begin{equation}
	\ell_N=0
	\qquad\Longrightarrow\qquad
	\dim\mathcal S_{\rm smooth}=2.
	\label{eq:fuchs-smooth-dim-two}
\end{equation}
We therefore obtain
\begin{equation}
	\boxed{
		\ell_N=0
		\quad\Longleftrightarrow\quad
		\dim\mathcal S_{\rm smooth}=2
	}
	\label{eq:fuchs-log-smooth-equivalence}
\end{equation}
.

\subsection{Local Fuchsian equivalence theorem}
\label{subsec:fuchs-main-theorem}

The preceding results can be summarized in a single local theorem.

\begin{theorem}[Local Fuchsian equivalence at a resonant horizon]
	\label{thm:fuchs-local-equivalence}
	Assume the nonextremal horizon conditions of Section~\ref{sec:setup} and suppose that, at exact integer resonance
	\begin{equation}
		\nu=N,
		\qquad
		N\in\mathbb N,
	\end{equation}
	the horizon radial equation can be written in the standard Fuchs form \eqref{eq:fuchs-operator}.
	Normalize the smaller-root branch by $a_0=1$, define $\Omega_N$ by Definition~\ref{def:fuchs-Omega}, and let $M_N$ be the coefficient matrix of the first $N$ recurrence equations.
	Then
	\begin{equation}
		\boxed{
			\det M_N=0
			\quad\Longleftrightarrow\quad
			\Omega_N=0
			\quad\Longleftrightarrow\quad
			\ell_N=0
			\quad\Longleftrightarrow\quad
			\dim\mathcal S_{\rm smooth}=2
		}
		\label{eq:fuchs-main-equivalence}
	\end{equation}
	.
	In the present standard Fuchs normalization,
	\begin{equation}
		\det M_N
		=
		[(N-1)!]^2\Omega_N,
		\qquad
		\ell_N
		=
		-\frac{\Omega_N}{N}.
		\label{eq:fuchs-main-identities}
	\end{equation}
\end{theorem}

\begin{proof}
	By Proposition~\ref{prop:fuchs-det-Omega},
	\begin{equation}
		\det M_N=0
		\Longleftrightarrow
		\Omega_N=0.
	\end{equation}
	Equation~\eqref{eq:fuchs-log-Omega} gives
	\begin{equation}
		\Omega_N=0
		\Longleftrightarrow
		\ell_N=0.
	\end{equation}
	Finally, Section~\ref{subsec:fuchs-smooth-space} gives
	\begin{equation}
		\ell_N=0
		\Longleftrightarrow
		\dim\mathcal S_{\rm smooth}=2.
	\end{equation}
	Combining these statements yields Eq.~\eqref{eq:fuchs-main-equivalence}.
\end{proof}

Theorem~\ref{thm:fuchs-local-equivalence} does not claim that integer resonance or logarithmic Frobenius theory itself is new.
Its role is to state, in a common notation and normalization, precisely how the rank loss of the finite recurrence matrix commonly used in pole-skipping analyses is related to the corresponding classical Fuchsian data.
In particular, the matrix condition
\begin{equation}
	\det M_N=0
\end{equation}
is not an independent criterion; it is the finite-dimensional representation of the solvability condition
\begin{equation}
	\Omega_N=0
\end{equation}
at the resonant order.

\subsection{Remarks on normalization}
\label{subsec:fuchs-normalization}

In the exact proportionality
\begin{equation}
	\det M_N
	=
	[(N-1)!]^2\Omega_N
\end{equation}
of Theorem~\ref{thm:fuchs-local-equivalence}, the factor $[(N-1)!]^2$ depends on the standard Fuchs normalization \eqref{eq:fuchs-operator} and on the row normalization of the recurrence equations.

For example, multiplying each recurrence equation by a nonzero constant rescales $\det M_N$ by the product of those constants.
Likewise, multiplying the entire radial equation by a nonvanishing analytic function near the horizon can change the exact normalization of the coefficient matrix.

Accordingly, in a more general normalization the invariant statement is
\begin{equation}
	\det M_N
	=
	\mathcal C_N\Omega_N,
	\qquad
	\mathcal C_N\neq0.
	\label{eq:fuchs-general-normalization}
\end{equation}
What matters in the present paper is the equivalence of the zero conditions,
\begin{equation}
	\boxed{
		\det M_N=0
		\quad\Longleftrightarrow\quad
		\Omega_N=0
	}.
\end{equation}

Similarly, if the original radial equation is written in a form such as $xR''+\cdots=0$ and is then multiplied by $x$ to obtain the standard Fuchs form \eqref{eq:fuchs-operator}, the power label of the coefficient equation can shift by one order.
This is only a change of normalization of the same resonant solvability condition and does not affect the condition
\begin{equation}
	\Omega_N=0.
\end{equation}

The local structure obtained in this section may be summarized as
\begin{equation}
	\boxed{
		\begin{gathered}
			\nu=N
			\quad\Longrightarrow\quad
			\text{the coefficient multiplying $a_N$ in the $n=N$ recurrence vanishes},
			\\[2mm]
			\Omega_N=0
			\quad\Longleftrightarrow\quad
			\text{log-free continuation of the smaller-root branch},
			\\[2mm]
			\Omega_N=0
			\quad\Longleftrightarrow\quad
			\det M_N=0
			\quad\Longleftrightarrow\quad
			\ell_N=0
			\quad\Longleftrightarrow\quad
			\dim\mathcal S_{\rm smooth}=2.
		\end{gathered}
	}
	\label{eq:fuchs-section-summary}
\end{equation}

\section{Horizon-to-boundary dictionary}
	\label{sec:dictionary}
	
	In Section~\ref{sec:fuchsian}, we analyzed the local Frobenius problem at integer resonance and obtained
	\begin{equation}
		\det M_N=0
		\quad\Longleftrightarrow\quad
		\Omega_N=0
		\quad\Longleftrightarrow\quad
		\ell_N=0
		\quad\Longleftrightarrow\quad
		\dim\mathcal S_{\rm smooth}=2
		\label{eq:dictionary-local-chain}
	\end{equation}
	.
	In this section, we connect this local result to the boundary asymptotic basis and connection problem defined in Section~\ref{sec:setup}.
	
	We first analyze the connection matrix between the horizon Frobenius basis and the boundary-normalized basis away from resonance.
	We then track directly from the recurrence how the nonresonant ingoing solution can become singular in parameter space as an integer resonance is approached.
	After removing this simple pole with the corresponding analytic factor, we compute the value of the regularized ingoing family at resonance and show that the same $\Omega_N$ simultaneously controls both regularized boundary connection coefficients.

	\subsection{Nonresonant horizon basis and boundary basis}
	\label{subsec:dictionary-nonresonant-basis}
	
	First consider a nonresonant parameter value satisfying
	\begin{equation}
		\nu\notin\mathbb Z.
		\label{eq:dictionary-nonresonant}
	\end{equation}
	We take the horizon Frobenius basis defined in Section~\ref{subsec:setup-horizon-branches} to be
	\begin{align}
		H_{\rm in}(x;\nu,\lambda)
		&=
		1+O(x),
		\label{eq:dictionary-Hin}
		\\
		H_{\rm out}(x;\nu,\lambda)
		&=
		x^\nu\bigl(1+O(x)\bigr)
		\label{eq:dictionary-Hout}
	\end{align}
	.
	
	The solution $H_{\rm in}$ is horizon-smooth in ingoing Eddington--Finkelstein coordinates.
	By contrast, when $\nu\notin\mathbb Z$, the noninteger-power behavior of $x^\nu$ implies generically that
	\begin{equation}
		H_{\rm out}\notin C^\infty
		\qquad
		\text{at }x=0.
		\label{eq:dictionary-Hout-nonsmooth}
	\end{equation}
	
	Let $R_1$ and $R_2$ be the two independent global solutions normalized with respect to the source and response channels at the boundary, as in Section~\ref{subsec:setup-boundary-basis}.
	Continuing them to the horizon gives
	\begin{align}
		R_1
		&=
		c_1H_{\rm in}+d_1H_{\rm out},
		\label{eq:dictionary-R1-horizon}
		\\
		R_2
		&=
		c_2H_{\rm in}+d_2H_{\rm out}.
		\label{eq:dictionary-R2-horizon}
	\end{align}
	where $c_i$ and $d_i$ are parameter-dependent connection coefficients.
	
	Define the connection matrix by
	\begin{equation}
		C
		=
		\begin{pmatrix}
			c_1 & c_2\\
			d_1 & d_2
		\end{pmatrix}
	\end{equation}
	and set
	\begin{equation}
		\Delta_c
		:=
		\det C
		=
		c_1d_2-c_2d_1
		\label{eq:dictionary-Delta-c}
	\end{equation}
	.
	If the two horizon solutions and the two boundary solutions are each linearly independent, then
	\begin{equation}
		\Delta_c\neq0.
		\label{eq:dictionary-Delta-nonzero}
	\end{equation}
	Indeed,
	\begin{equation}
		W[R_1,R_2]
		=
		\Delta_c\,
		W[H_{\rm in},H_{\rm out}]
		\label{eq:dictionary-Wronskian}
	\end{equation}
	so independence of the two bases is equivalent to invertibility of the connection matrix.
	
	Solving Eqs.~\eqref{eq:dictionary-R1-horizon} and \eqref{eq:dictionary-R2-horizon} for the horizon basis gives
	\begin{align}
		H_{\rm in}
		&=
		\frac{d_2}{\Delta_c}R_1
		-
		\frac{d_1}{\Delta_c}R_2,
		\label{eq:dictionary-Hin-boundary}
		\\
		H_{\rm out}
		&=
		-\frac{c_2}{\Delta_c}R_1
		+
		\frac{c_1}{\Delta_c}R_2.
		\label{eq:dictionary-Hout-boundary}
	\end{align}
	
	Using the convention of Section~\ref{subsec:setup-boundary-basis}, write the boundary decomposition of the ingoing solution as
	\begin{equation}
		H_{\rm in}
		=
		A R_1+B R_2
		\label{eq:dictionary-Hin-AB}
	\end{equation}
	. Equation~\eqref{eq:dictionary-Hin-boundary} then gives
	\begin{equation}
		A
		=
		\frac{d_2}{\Delta_c},
		\qquad
		B
		=
		-\frac{d_1}{\Delta_c}.
		\label{eq:dictionary-AB-d}
	\end{equation}
	Thus the source coefficient $A$ is related to the coefficient $d_2$ of $H_{\rm out}$ in $R_2$, while the response coefficient $B$ is related to the coefficient $d_1$ of $H_{\rm out}$ in $R_1$.

	\subsection{Zero conditions for the source and response}
	\label{subsec:dictionary-nonresonant-zero}
	
	In Eq.~\eqref{eq:dictionary-R1-horizon}, $H_{\rm in}$ is horizon-smooth whereas $H_{\rm out}$ is not horizon-smooth for $\nu\notin\mathbb Z$.
	Therefore
	\begin{equation}
		R_1\in C^\infty
		\qquad\Longleftrightarrow\qquad
		d_1=0
		\label{eq:dictionary-R1-smooth}
	\end{equation}
	.
	Similarly,
	\begin{equation}
		R_2\in C^\infty
		\qquad\Longleftrightarrow\qquad
		d_2=0.
		\label{eq:dictionary-R2-smooth}
	\end{equation}
	
	Using Eq.~\eqref{eq:dictionary-AB-d} together with $\Delta_c\neq0$ gives the following result.
	
	\begin{proposition}[Nonresonant horizon-to-boundary relation]
		\label{prop:dictionary-nonresonant}
		Suppose $\nu\notin\mathbb Z$, with $R_1,R_2$ linearly independent and $H_{\rm in},H_{\rm out}$ linearly independent.
		Then
		\begin{equation}
			\boxed{
				A=0
				\quad\Longleftrightarrow\quad
				R_2\in C^\infty
				\text{ at the horizon}
			}
			\label{eq:dictionary-A-R2}
		\end{equation}
		and
		\begin{equation}
			\boxed{
				B=0
				\quad\Longleftrightarrow\quad
				R_1\in C^\infty
				\text{ at the horizon}
			}
			\label{eq:dictionary-B-R1}
		\end{equation}
		.
	\end{proposition}
	
	\begin{proof}
		Equation~\eqref{eq:dictionary-AB-d} gives
		\begin{equation}
			A=0
			\quad\Longleftrightarrow\quad
			d_2=0,
		\end{equation}
		and Eq.~\eqref{eq:dictionary-R2-smooth} shows that
		$d_2=0$ is equivalent to horizon smoothness of $R_2$.
		The second equivalence follows in the same way.
	\end{proof}
	
	The crossed appearance of source and response in Eqs.~\eqref{eq:dictionary-A-R2} and \eqref{eq:dictionary-B-R1} has a simple origin.
	Although $A$ is the coefficient of $R_1$ in the boundary expansion of $H_{\rm in}$, Eq.~\eqref{eq:dictionary-Hin-boundary} shows that its value is determined by $d_2$, the nonsmooth horizon component of $R_2$.
	The analogous crossed relation holds for $B$.
	
	Away from resonance, if $H_{\rm in}$ is rescaled as
	\begin{equation}
		H_{\rm in}
		\longrightarrow
		f(\nu,\lambda)H_{\rm in},
		\qquad
		f(\nu,\lambda)\neq0
	\end{equation}
	then
	\begin{equation}
		(A,B)
		\longrightarrow
		f(\nu,\lambda)(A,B).
	\end{equation}
	Therefore, if $f$ is finite and nonzero in the parameter neighborhood of interest, the zero conditions in Eqs.~\eqref{eq:dictionary-A-R2} and \eqref{eq:dictionary-B-R1} are unchanged.

	\subsection{Parameter-dependent recurrence near integer resonance}
	\label{subsec:dictionary-parameter-recurrence}
	
	We now consider a neighborhood of an integer resonance
	\begin{equation}
		\nu=N,
		\qquad
		N\in\mathbb N
		\label{eq:dictionary-resonance}
	\end{equation}
	.
	We take the parameter neighborhood sufficiently small that it contains no other integer resonance.
	
	We abbreviate the parameters other than the frequency by
	\begin{equation}
		\lambda=(k,\Delta,\mathbf g,\ldots)
	\end{equation}
	and normalize the nonresonant exponent-zero ingoing solution as
	\begin{equation}
		H_{\rm in}(x;\nu,\lambda)
		=
		\sum_{n=0}^{\infty}
		h_n(\nu,\lambda)x^n,
		\qquad
		h_0=1
		\label{eq:dictionary-Hin-series}
	\end{equation}
	.
	
	Keeping the recurrence of Section~\ref{subsec:fuchs-standard} in parameter-dependent form near resonance, we write
	\begin{equation}
		D_n(\nu)h_n
		+
		\sum_{m=0}^{n-1}
		K_{n,m}(\nu,\lambda)h_m
		=
		0,
		\label{eq:dictionary-general-recurrence}
	\end{equation}
	where
	\begin{equation}
		D_n(\nu)
		=
		n(n-\nu)
		\label{eq:dictionary-Dn}
	\end{equation}
	and assume that $K_{n,m}$ is analytic in the neighborhood of interest.
	In the standard Fuchs form,
	\begin{equation}
		K_{n,m}
		=
		mp_{n-m}+q_{n-m}.
	\end{equation}
	
	For $1\leq n<N$,
	\begin{equation}
		D_n(N)
		=
		n(n-N)
		\neq0.
	\end{equation}
	so solving the recurrence successively starting from $h_0=1$ shows that
	\begin{equation}
		h_1,\ldots,h_{N-1}
	\end{equation}
	are analytic near $\nu=N$.
	
	After substituting these lower coefficients into the recurrence at $n=N$, define the remaining quantity by
	\begin{equation}
		\Omega_N(\nu,\lambda)
		:=
		\sum_{m=0}^{N-1}
		K_{N,m}(\nu,\lambda)h_m(\nu,\lambda)
		\label{eq:dictionary-Omega-parameter}
	\end{equation}
	.
	At $\nu=N$, this definition agrees with the $\Omega_N$ of Section~\ref{sec:fuchsian}.
	
	The $n=N$ recurrence therefore becomes
	\begin{equation}
		N(N-\nu)h_N
		+
		\Omega_N(\nu,\lambda)
		=
		0
		\label{eq:dictionary-N-recurrence}
	\end{equation}
	and for $\nu\neq N$ we obtain
	\begin{equation}
		\boxed{
			h_N(\nu,\lambda)
			=
			\frac{\Omega_N(\nu,\lambda)}
			{N(\nu-N)}
		}
		\label{eq:dictionary-hN}
	\end{equation}
	.
	
	Hence, if
	\begin{equation}
		\Omega_N(N,\lambda)\neq0
		\label{eq:dictionary-generic-resonance}
	\end{equation}
	at a generic integer resonance, $h_N$ has a genuine simple pole at $\nu=N$.
	By contrast, if
	\begin{equation}
		\Omega_N(N,\lambda_*)=0
		\label{eq:dictionary-compatible-resonance}
	\end{equation}
	then the numerator and denominator vanish simultaneously.
	 
\subsection{Higher coefficients and the simple parameter pole}
\label{subsec:dictionary-higher-coefficients}

For $n>N$,
\begin{equation}
	D_n(N)
	=
	n(n-N)
	\neq0.
	\label{eq:dictionary-higher-pivot}
\end{equation}
The $n=N+1$ recurrence is therefore
\begin{equation}
	D_{N+1}h_{N+1}
	+
	\sum_{m=0}^{N-1}
	K_{N+1,m}h_m
	+
	K_{N+1,N}h_N
	=
	0
\end{equation}
and hence
\begin{equation}
	h_{N+1}
	=
	-\frac{1}{D_{N+1}}
	\sum_{m=0}^{N-1}
	K_{N+1,m}h_m
	-
	\frac{K_{N+1,N}}{D_{N+1}}h_N.
	\label{eq:dictionary-hNp1-explicit}
\end{equation}
Define
\begin{equation}
	P_{N+1}
	:=
	-\frac{1}{D_{N+1}}
	\sum_{m=0}^{N-1}
	K_{N+1,m}h_m,
	\qquad
	Q_{N+1}
	:=
	-\frac{K_{N+1,N}}{D_{N+1}}
	\label{eq:dictionary-PQ-Np1}
\end{equation}
Then
\begin{equation}
	h_{N+1}
	=
	P_{N+1}
	+
	Q_{N+1}h_N.
	\label{eq:dictionary-hNp1}
\end{equation}

As shown in Section~\ref{subsec:dictionary-parameter-recurrence}, $h_0,\ldots,h_{N-1}$ are analytic near $\nu=N$, and $K_{N+1,m}$ is analytic there as well.
Moreover,
\begin{equation}
	D_{N+1}(N)
	=
	(N+1)\bigl((N+1)-N\bigr)
	=
	N+1
	\neq0
\end{equation}
and hence
$P_{N+1}$ and $Q_{N+1}$ are analytic near $\nu=N$.

Substituting Eq.~\eqref{eq:dictionary-hNp1} into the next recurrence gives
\begin{equation}
	h_{N+2}
	=
	P_{N+2}
	+
	Q_{N+2}h_N
\end{equation}
and repeating the same argument yields inductively
\begin{equation}
	\boxed{
		h_{N+j}
		=
		P_{N+j}(\nu,\lambda)
		+
		Q_{N+j}(\nu,\lambda)h_N,
		\qquad
		j\geq1,
	}
	\label{eq:dictionary-affine-tail}
\end{equation}
where all $P_{N+j}$ and $Q_{N+j}$ are
analytic near $\nu=N$.

Therefore, in coefficients above $h_N$,
no new higher powers such as $h_N^2,h_N^3,\ldots$ are generated.
All possible singular parts are transmitted through the single resonant coefficient $h_N$.

Combined with Eq.~\eqref{eq:dictionary-hN}, this implies that
\begin{equation}
	h_n(\nu,\lambda),
	\qquad
	n\geq N,
\end{equation}
can have at most a simple pole at $\nu=N$.

If a particular $Q_{N+j}(N,\lambda)$ vanishes, the pole residue of that coefficient may vanish as well; the general statement is therefore ``at most a simple pole,'' not necessarily ``a simple pole.''
At a generic resonance with $\Omega_N(N,\lambda)\neq0$, however, $h_N$ itself has a genuine simple pole. 
	
	\subsection{Regularized ingoing family}
	\label{subsec:dictionary-regularization}
	
	To remove this simple pole without introducing an additional zero at resonance, choose an analytic function near $\nu=N$,
	\begin{equation}
		g_N(\nu,\lambda)
		=
		(\nu-N)u_N(\nu,\lambda),
		\qquad
		u_N(N,\lambda)\neq0
		\label{eq:dictionary-regularizer}
	\end{equation}
	so that $g_N$ has exactly one simple zero at $\nu=N$.
	
	Using this factor, define the regularized ingoing family by
	\begin{equation}
		\boxed{
			H_{\rm in}^{\rm reg}
			:=
			g_N H_{\rm in}
		}
		\label{eq:dictionary-Hreg}
	\end{equation}
	 .
	
	Define the coefficients
	\begin{equation}
		\widehat h_n
		:=
		g_Nh_n
		\label{eq:dictionary-hhat}
	\end{equation}
	so that
	\begin{equation}
		H_{\rm in}^{\rm reg}
		=
		\sum_{n=0}^{\infty}
		\widehat h_nx^n.
	\end{equation}
	
	For $n<N$, $h_n$ is finite at $\nu=N$, while $g_N(N,\lambda)=0$. Hence
	\begin{equation}
		\widehat h_0(N,\lambda)
		=
		\cdots
		=
		\widehat h_{N-1}(N,\lambda)
		=
		0.
		\label{eq:dictionary-lower-vanish}
	\end{equation}
	
	At $n=N$, by contrast, Eqs.~\eqref{eq:dictionary-hN} and \eqref{eq:dictionary-regularizer} give
	\begin{align}
		\widehat h_N
		&=
		g_Nh_N
		\nonumber\\
		&=
		(\nu-N)u_N
		\frac{\Omega_N}{N(\nu-N)}
		\nonumber\\
		&=
		\frac{u_N}{N}\Omega_N.
	\end{align}
	Therefore
	\begin{equation}
		\boxed{
			\widehat h_N(N,\lambda)
			=
			\frac{u_N(N,\lambda)}{N}
			\Omega_N(N,\lambda)
		}
		\label{eq:dictionary-hhatN}
	\end{equation}
	 .
	
	Multiplying Eq.~\eqref{eq:dictionary-affine-tail} by $g_N$ also gives
	\begin{equation}
		\widehat h_{N+j}
		=
		g_NP_{N+j}
		+
		Q_{N+j}\widehat h_N.
		\label{eq:dictionary-hhat-tail}
	\end{equation}
	At $\nu=N$, the first term vanishes, so
	\begin{equation}
		\widehat h_{N+j}(N,\lambda)
		=
		Q_{N+j}(N,\lambda)
		\widehat h_N(N,\lambda),
		\qquad
		j\geq1.
		\label{eq:dictionary-hhat-tail-resonance}
	\end{equation}
	
	Thus, at exact resonance, the entire Taylor tail of the regularized ingoing family is determined by the single coefficient $\widehat h_N(N,\lambda)$.

	\subsection{Identity at exact resonance}
	\label{subsec:dictionary-resonance-identity}
	
	At exact resonance, normalize the Frobenius solution associated with the larger indicial root $N$ as
	\begin{equation}
		H_{\rm out}^{(N)}(x;\lambda)
		=
		x^N
		\left(
		1+b_1x+b_2x^2+\cdots
		\right)
		\label{eq:dictionary-HoutN}
	\end{equation}
	Its Taylor coefficients are
	\begin{equation}
		0,\ldots,0,
		\underbrace{1}_{n=N},
		b_1,b_2,\ldots
	\end{equation}
	 .
	
	For $n>N$, all recurrence coefficients $D_n(\nu)$ multiplying $h_n$ are nonzero, and hence
	the higher coefficients are uniquely determined once the $n=N$ coefficient is specified.
	On the other hand, Eqs.~\eqref{eq:dictionary-lower-vanish},
	\eqref{eq:dictionary-hhatN},
	and \eqref{eq:dictionary-hhat-tail-resonance} show that
	$H_{\rm in}^{\rm reg}(N,\lambda)$ also has
	all coefficients below $n=N$ equal to zero, and
	its entire higher-order tail is determined by the single coefficient at $n=N$.
	
	Therefore, uniqueness of the recurrence gives the following identity.
	
	\begin{theorem}[Identity for the regularized ingoing family at exact resonance]
		\label{thm:dictionary-resonance-identity}
		Assume the local parameter conditions of Section~\ref{subsec:setup-global-assumptions}.
		Then, at exact integer resonance $\nu=N$,
		\begin{equation}
			\boxed{
				H_{\rm in}^{\rm reg}(x;N,\lambda)
				=
				\frac{u_N(N,\lambda)}{N}
				\Omega_N(N,\lambda)
				H_{\rm out}^{(N)}(x;\lambda)
			}
			\label{eq:dictionary-resonance-identity}
		\end{equation}
		holds.
	\end{theorem}
	
	\begin{proof}
		All Taylor coefficients of $H_{\rm in}^{\rm reg}(N,\lambda)$ below $n=N$ vanish, while the coefficient at $n=N$ is
		\begin{equation}
			\frac{u_N(N,\lambda)}{N}
			\Omega_N(N,\lambda)
		\end{equation}
		 .
		The solution $H_{\rm out}^{(N)}$ satisfies the same differential equation at resonance and
		is normalized to have coefficient $1$ at $n=N$.
		For $n>N$, the recurrence is nondegenerate, so uniqueness of the higher coefficients
		implies Eq.~\eqref{eq:dictionary-resonance-identity}.
	\end{proof} 
The recurrence argument above shows that
the two sides of Eq.~\eqref{eq:dictionary-resonance-identity}
have identical Taylor coefficients at $x=0$.
To interpret this as an identity between actual analytic solutions,
one must verify that the corresponding Frobenius series converge in a common neighborhood of the horizon.
Appendix~\ref{app:frobenius-convergence}
uses a majorant argument for the parameter-dependent recurrence to establish
parameter analyticity of the regularized Frobenius coefficients and
local uniform convergence in a common $x$-disk.
Therefore, Eq.~\eqref{eq:dictionary-resonance-identity} is
an identity not only of formal power series but also of actual local analytic solutions.

	\subsection{Boundary connection coefficients at resonance}
	\label{subsec:dictionary-boundary-resonance}
	
	We now continue Theorem~\ref{thm:dictionary-resonance-identity} to the boundary.
	We use the continuation path fixed in Section~\ref{subsec:setup-global-assumptions} and assume that, near resonance, the boundary basis $R_1,R_2$ can be chosen to depend analytically on the parameters.
	
	Define the boundary decomposition of the regularized ingoing family by
	\begin{equation}
		H_{\rm in}^{\rm reg}
		=
		\widehat A(\nu,\lambda)R_1
		+
		\widehat B(\nu,\lambda)R_2
		\label{eq:dictionary-Hreg-boundary}
	\end{equation}
	.
	As specified in Section~\ref{subsec:setup-global-assumptions},
	$\widehat A$ and $\widehat B$ are
	regularized connection coefficient functions distinct from the ordinary nonresonant coefficients $A,B$.
	
	Write the larger-root solution at resonance in the same boundary basis as
	\begin{equation}
		H_{\rm out}^{(N)}
		=
		A_{\rm out}R_1
		+
		B_{\rm out}R_2
		\label{eq:dictionary-HoutN-boundary}
	\end{equation}
	Since $H_{\rm out}^{(N)}$ is a nonzero solution and $R_1,R_2$ are independent,
	\begin{equation}
		(A_{\rm out},B_{\rm out})
		\neq
		(0,0).
		\label{eq:dictionary-AoutBout}
	\end{equation}
	
	Substituting Theorem~\ref{thm:dictionary-resonance-identity} into Eq.~\eqref{eq:dictionary-HoutN-boundary} and comparing the coefficients of $R_1$ and $R_2$ gives
	\begin{align}
		\widehat A(N,\lambda)
		&=
		\frac{u_N(N,\lambda)}{N}
		\Omega_N(N,\lambda)A_{\rm out},
		\label{eq:dictionary-Ahat}
		\\
		\widehat B(N,\lambda)
		&=
		\frac{u_N(N,\lambda)}{N}
		\Omega_N(N,\lambda)B_{\rm out}.
		\label{eq:dictionary-Bhat}
	\end{align}
	
	Therefore, we obtain the following global statement.
	
	\begin{theorem}[Horizon-to-boundary connection theorem at resonance]
		\label{thm:dictionary-global}
		Assume the standing assumptions of Section~\ref{subsec:setup-global-assumptions}, and suppose that near exact resonance $\nu=N$ the boundary basis $R_1,R_2$ can be chosen to depend analytically on the parameters, with
		\begin{equation}
			W[R_1,R_2]\neq0
		\end{equation}
		Then
		\begin{equation}
			\boxed{
				\Omega_N(N,\lambda)=0
				\quad\Longleftrightarrow\quad
				\widehat A(N,\lambda)
				=
				\widehat B(N,\lambda)
				=
				0
			}
			\label{eq:dictionary-Omega-ABhat}
		\end{equation}
		holds.
	\end{theorem}
	
	\begin{proof}
		If $\Omega_N=0$, Eqs.~\eqref{eq:dictionary-Ahat} and
		\eqref{eq:dictionary-Bhat} give
		$\widehat A=\widehat B=0$ .
		
		Conversely, suppose that
		\begin{equation}
			\widehat A=\widehat B=0.
		\end{equation}
		Since $u_N(N,\lambda)\neq0$, Eqs.~\eqref{eq:dictionary-Ahat} and
		\eqref{eq:dictionary-Bhat} imply
		\begin{equation}
			\Omega_N A_{\rm out}=0,
			\qquad
			\Omega_N B_{\rm out}=0.
		\end{equation}
		Since $(A_{\rm out},B_{\rm out})\neq(0,0)$, it follows that
		\begin{equation}
			\Omega_N=0.
		\end{equation}
	\end{proof}
	
	Combining this theorem with Theorem~\ref{thm:fuchs-local-equivalence} gives
	\begin{equation}
		\boxed{
			\begin{gathered}
				\det M_N=0
				\quad\Longleftrightarrow\quad
				\Omega_N=0
				\quad\Longleftrightarrow\quad
				\ell_N=0
				\quad\Longleftrightarrow\quad
				\dim\mathcal S_{\rm smooth}=2
				\\[1mm]
				\Longleftrightarrow\quad
				\widehat A=\widehat B=0
			\end{gathered}
		}
		\label{eq:dictionary-full-chain}
	\end{equation}
	.
	
	Moreover, if $R_1,R_2$ can be continued to the horizon and $W[R_1,R_2]\neq0$, then
	\begin{equation}
		\boxed{
			\Omega_N=0
			\quad\Longleftrightarrow\quad
			R_1,R_2\in C^\infty
			\text{ at the horizon}
		}
		\label{eq:dictionary-simultaneous-smoothness}
	\end{equation}
	 .
	
	Indeed, if $\Omega_N=0$,
	Theorem~\ref{thm:fuchs-local-equivalence} implies that
	the local smooth solution space is two-dimensional, and hence every local solution is
	a linear combination of two smooth Frobenius solutions.
	Therefore both $R_1$ and $R_2$ are horizon-smooth.
	
	Conversely, if both $R_1$ and $R_2$ are horizon-smooth and
	$W[R_1,R_2]\neq0$, then there exist two independent smooth local solutions.
	Therefore
	\begin{equation}
		\dim\mathcal S_{\rm smooth}=2
	\end{equation}
	and Theorem~\ref{thm:fuchs-local-equivalence} again implies that
	$\Omega_N=0$ .
	
	Therefore, at exact resonance,
	\begin{equation}
		\boxed{
			R_1,R_2\in C^\infty
			\quad\Longleftrightarrow\quad
			\widehat A=\widehat B=0
		}
		\label{eq:dictionary-smooth-zero}
	\end{equation}
	is the resulting global dictionary.
	The connection between the two statements is provided by
	Theorem~\ref{thm:fuchs-local-equivalence} and
	Theorem~\ref{thm:dictionary-global}.

	\subsection{Two distinct statements at exact resonance}
	\label{subsec:dictionary-distinction}
	
	Finally, to avoid conflating the nonresonant and resonant results, we make two distinctions explicit.
	
	First, Proposition~\ref{prop:dictionary-nonresonant},
	\begin{equation}
		A=0
		\Longleftrightarrow
		R_2\in C^\infty,
		\qquad
		B=0
		\Longleftrightarrow
		R_1\in C^\infty
		\label{eq:dictionary-nonresonant-repeat}
	\end{equation}
	holds in the region
	\begin{equation}
		\nu\notin\mathbb Z.
	\end{equation}
	Its proof uses the fact that $H_{\rm out}\notin C^\infty$.
	
	At exact resonance $\nu=N$, however,
	\begin{equation}
		H_{\rm out}^{(N)}
		=
		x^N(1+O(x))
	\end{equation}
	is itself analytic and horizon-smooth.
	Therefore Eq.~\eqref{eq:dictionary-nonresonant-repeat} cannot be applied literally at exact resonance to infer $\widehat A=\widehat B=0$.
	
	The simultaneous zero at exact resonance,
	\begin{equation}
		\widehat A=\widehat B=0
	\end{equation}
	instead follows directly from Theorem~\ref{thm:dictionary-resonance-identity},
	\begin{equation}
		H_{\rm in}^{\rm reg}(N,\lambda)
		=
		\frac{u_N(N,\lambda)}{N}
		\Omega_N(N,\lambda)
		H_{\rm out}^{(N)}
	\end{equation}
	 .
	
	Second,
	\begin{equation}
		\widehat A(P_*)=\widehat B(P_*)=0
		\label{eq:dictionary-hat-zero}
	\end{equation}
	does not mean that a single nonzero resonant solution has both boundary coefficients equal to zero.
	
	If $R_1,R_2$ are independent, a nonzero solution of the ODE at resonance must instead be written as
	\begin{equation}
		R_{\rm res}
		=
		A_*R_1+B_*R_2,
		\qquad
		(A_*,B_*)\neq(0,0)
		\label{eq:dictionary-actual-resonant}
	\end{equation}
	 .
	
	On the other hand, when $\Omega_N(P_*)=0$, Theorem~\ref{thm:dictionary-resonance-identity} gives
	\begin{equation}
		H_{\rm in}^{\rm reg}(P_*)
		=
		0.
		\label{eq:dictionary-Hreg-zero}
	\end{equation}
	Therefore
	\begin{equation}
		\widehat A(P_*)=\widehat B(P_*)=0
	\end{equation}
	expresses the fact that the parameter-dependent regularized ingoing family itself
	has the zero solution as its value at $P_*$, as expressed in the boundary basis.
	This zero coefficient pair does not represent the projective class $[A:B]$ defined by the boundary coefficient pair of
	a nonzero resonant solution.
	
	The resonant differential equation itself
	still has two independent smooth solutions when $\Omega_N=0$.
	Which nonzero linear combination of those solutions
	is selected by a particular parameter-space limit of the nonresonant ingoing solution
	is a separate question.
	This question is addressed in Section~\ref{sec:resonant-limit}.

	\subsection{Independence of the factor removing the simple parameter pole}
	\label{subsec:dictionary-reg-normalization}
	
	The analytic factor $g_N$ chosen in Eq.~\eqref{eq:dictionary-regularizer} to remove the simple parameter pole is not unique.
	Consider two such factors with the same simple zero,
	\begin{equation}
		g_N
		=
		(\nu-N)u_N,
		\qquad
		\widetilde g_N
		=
		(\nu-N)\widetilde u_N
	\end{equation}
	with
	\begin{equation}
		u_N(N,\lambda)\neq0,
		\qquad
		\widetilde u_N(N,\lambda)\neq0.
	\end{equation}
	Then
	\begin{equation}
		f(\nu,\lambda)
		:=
		\frac{\widetilde g_N}{g_N}
		=
		\frac{\widetilde u_N}{u_N}
		\label{eq:dictionary-f-reg}
	\end{equation}
	is analytic near resonance and satisfies
	\begin{equation}
		f(N,\lambda)\neq0.
	\end{equation}
	
	Therefore
	\begin{equation}
		\widetilde H_{\rm in}^{\rm reg}
		=
		fH_{\rm in}^{\rm reg}
	\end{equation}
	and the corresponding boundary coefficients obey
	\begin{equation}
		(\widetilde{\widehat A},
		\widetilde{\widehat B})
		=
		f(\widehat A,\widehat B).
	\end{equation}
	Since $f(N,\lambda)\neq0$,
	\begin{equation}
		\widehat A=\widehat B=0
		\quad\Longleftrightarrow\quad
		\widetilde{\widehat A}
		=
		\widetilde{\widehat B}
		=
		0.
		\label{eq:dictionary-reg-invariance}
	\end{equation} 
	Therefore, the simultaneous-zero condition is identical for
	all choices that remove the generic simple parameter pole with a single zero.
	The relevant requirement is that the pole-removing factor have a simple zero, and not a higher-order zero, at $\nu=N$.
	For example, multiplication by $(\nu-N)^2$ introduces an additional zero beyond what is needed to remove the $1/(\nu-N)$ pole. Even at a generic resonance with $\Omega_N(N,\lambda)\neq0$, such a choice would force the regularized family to vanish at resonance and could therefore create a spurious simultaneous zero unrelated to $\Omega_N=0$. Such choices are excluded.
	
	The results of this section may be summarized as follows.
	Away from resonance,
	\begin{equation}
		\boxed{
			A=0
			\Longleftrightarrow
			R_2\in C^\infty,
			\qquad
			B=0
			\Longleftrightarrow
			R_1\in C^\infty.
		}
		\label{eq:dictionary-summary-nonres}
	\end{equation}
	At integer resonance, the same local solvability function controls the global boundary data through
	\begin{equation}
		\boxed{
			H_{\rm in}^{\rm reg}(N,\lambda)
			=
			\frac{u_N(N,\lambda)}{N}
			\Omega_N(N,\lambda)
			H_{\rm out}^{(N)}
		}
		\label{eq:dictionary-summary-identity}
	\end{equation}
	and
	\begin{equation}
		\boxed{
			\Omega_N=0
			\quad\Longleftrightarrow\quad
			\widehat A=\widehat B=0.
		}
		\label{eq:dictionary-summary-zero}
	\end{equation}
	
	Therefore local Fuchsian resonance and the simultaneous source--response zero at the boundary
	are not independent phenomena;
	when the parameter-dependent ingoing family is followed to resonance,
	they are linked by the same $\Omega_N$.

	 \section{Resonant limits of the ingoing solution}
	 \label{sec:resonant-limit}
	 
Section~\ref{sec:dictionary} showed that, at integer resonance, the regularized ingoing family satisfies
\begin{equation}
    H_{\rm in}^{\rm reg}(N,\lambda)
    =
    \frac{u_N(N,\lambda)}{N}
    \Omega_N(N,\lambda)
    H_{\rm out}^{(N)}
\end{equation}
and that
\begin{equation}
    \Omega_N=0
    \quad\Longleftrightarrow\quad
    \widehat A=\widehat B=0.
\end{equation}
	 
At $\Omega_N=0$, however, the value of $H_{\rm in}^{\rm reg}$ at the resonant point is the zero solution. This does not mean that the nonzero solutions of the differential equation at resonance disappear. By Theorem~\ref{thm:fuchs-local-equivalence}, two independent smooth horizon solutions exist at this point.
	 
A distinct question remains: if the nonresonant ingoing solution normalized by $h_0=1$ is taken to exact resonance along a specified curve in parameter space, which nonzero element of the two-dimensional smooth solution space is selected?
	 
We analyze this question directly from the recurrence at $n=N$. The key point is that the ratio of the first-order variation of the solvability function $\Omega_N$ to that of the recurrence coefficient $D_N$, which multiplies the $N$th Frobenius coefficient, determines the resonant Frobenius coefficient selected in the limit.

\subsection{Free Frobenius coefficient at exact resonance}
\label{subsec:resonant-free-coordinate}
	 
Let $P$ denote a point in parameter space, and let the resonant point be
\begin{equation}
    P_*=(\omega_N,\lambda_*).
\end{equation}
Thus
\begin{equation}
    \nu(P_*)=N,
    \qquad
    \Omega_N(P_*)=0.
    \label{eq:resonant-point}
\end{equation}
	 
Using the notation of Section~\ref{subsec:dictionary-parameter-recurrence}, the recurrence at $n=N$ takes the form, near resonance,
\begin{equation}
    D_N(P)h_N(P)+\Omega_N(P)=0,
    \label{eq:resonant-N-recurrence}
\end{equation}
where
\begin{equation}
    D_N(P):=N\bigl(N-\nu(P)\bigr).
    \label{eq:resonant-DN}
\end{equation}
	 
At exact resonance,
\begin{equation}
    D_N(P_*)=0,
    \qquad
    \Omega_N(P_*)=0,
\end{equation}
so Eq.~\eqref{eq:resonant-N-recurrence} reduces to
\begin{equation}
    0\cdot h_N+0=0.
\end{equation}
Hence $h_N$ is no longer determined by the recurrence.
	 
For $n<N$,
\begin{equation}
    D_n(P_*)=n(n-N)\neq0,
\end{equation}
so once $h_0=1$ is fixed, the coefficients
\begin{equation}
    h_1^*,\ldots,h_{N-1}^*
\end{equation}
are uniquely determined. Here
\begin{equation}
    h_n^*:=h_n(P_*)=h_n(N,\lambda_*),
    \qquad 1\leq n<N,
\end{equation}
denotes the lower Frobenius coefficients at the resonant point $P_*$.
	 
At exact resonance, define the smaller-root solution for which the coefficient of $x^N$ is chosen to vanish by
\begin{equation}
    H_0^{(N)}(x)
    =
    1+h_1^*x+\cdots+h_{N-1}^*x^{N-1}
    +0\cdot x^N+O(x^{N+1}).
    \label{eq:resonant-H0}
\end{equation}
	 
Because $\Omega_N(P_*)=0$, Theorem~\ref{thm:fuchs-local-equivalence} implies that this solution is analytic and contains no logarithmic term.
	 
The larger-root Frobenius solution is normalized, as in Section~\ref{subsec:dictionary-resonance-identity}, by
\begin{equation}
    H_{\rm out}^{(N)}(x)
    =
    x^N\bigl(1+b_1^*x+b_2^*x^2+\cdots\bigr).
    \label{eq:resonant-Hout}
\end{equation}
	 
Since the two solutions begin with $x^0$ and $x^N$, respectively, they are linearly independent. Therefore every resonant solution satisfying the normalization $h_0=1$ can be written uniquely as
\begin{equation}
    \boxed{
        R_{\rm res}(x)
        =
        H_0^{(N)}(x)
        +
        \alpha_N H_{\rm out}^{(N)}(x)
    }
    \label{eq:resonant-general-solution}
\end{equation}
	 
Here $\alpha_N$ is the Frobenius coefficient at order $n=N$ that remains free at exact resonance. Once the normalizations in Eqs.~\eqref{eq:resonant-H0} and \eqref{eq:resonant-Hout} are fixed, the meaning of $\alpha_N$ is fixed as well.
	 
At exact resonance, both $H_0^{(N)}$ and $H_{\rm out}^{(N)}$ are horizon-smooth. Horizon smoothness alone therefore does not determine $\alpha_N$. The value actually selected is determined by how the nonresonant ingoing family approaches $P_*$.
	 
The solution $H_{\rm out}^{(N)}$ is also naturally related to the resonant limit of the second nonresonant Frobenius solution. Away from resonance,
\begin{equation}
    H_{\rm out}(x;P)
    =
    x^{\nu(P)}\bigl(1+O(x)\bigr),
\end{equation}
and as $\nu(P)\to N$ its leading behavior tends to
\begin{equation}
    x^{\nu(P)}\longrightarrow x^N.
\end{equation}
Thus $\alpha_N$ may be interpreted as the coefficient measuring the component that comes from the second nonresonant Frobenius solution in the resonant limit.

\subsection{General one-parameter continuation}
\label{subsec:resonant-general-curve}
	 
Consider a $C^1$ curve in parameter space approaching $P_*$,
\begin{equation}
    \gamma(s),
    \qquad
    \gamma(0)=P_*.
    \label{eq:resonant-gamma}
\end{equation}
Assume that for sufficiently small $s\neq0$, the punctured neighborhood lies in the nonresonant region.
	 
Normalize the nonresonant ingoing solution by $h_0=1$. Equation~\eqref{eq:resonant-N-recurrence} gives
\begin{equation}
    h_N\bigl(\gamma(s)\bigr)
    =
    -\frac{\Omega_N\bigl(\gamma(s)\bigr)}{D_N\bigl(\gamma(s)\bigr)}.
    \label{eq:resonant-hN-ratio}
\end{equation}
	 
At the resonant point, both the numerator and the denominator vanish. To determine the first-order resonant limit, assume that the curve crosses the resonance surface nondegenerately to first order:
\begin{equation}
    \left.
    \frac{d}{ds}D_N\bigl(\gamma(s)\bigr)
    \right|_{s=0}\neq0.
    \label{eq:resonant-transverse-D}
\end{equation}
	 
Since $D_N\circ\gamma$ and $\Omega_N\circ\gamma$ are $C^1$,
\begin{align}
    D_N\bigl(\gamma(s)\bigr)
    &=
    s\left.
    \frac{d}{ds}D_N\bigl(\gamma(s)\bigr)
    \right|_{s=0}+o(s),
    \label{eq:resonant-D-expand}
    \\
    \Omega_N\bigl(\gamma(s)\bigr)
    &=
    s\left.
    \frac{d}{ds}\Omega_N\bigl(\gamma(s)\bigr)
    \right|_{s=0}+o(s).
    \label{eq:resonant-Omega-expand}
\end{align}
	 
Therefore the limit in Eq.~\eqref{eq:resonant-hN-ratio} exists and is
\begin{equation}
    \boxed{
        \alpha_N^{(\gamma)}
        :=
        \lim_{s\to0}h_N\bigl(\gamma(s)\bigr)
        =
        -
        \frac{
            \left.\dfrac{d}{ds}\Omega_N\bigl(\gamma(s)\bigr)\right|_{s=0}
        }{
            \left.\dfrac{d}{ds}D_N\bigl(\gamma(s)\bigr)\right|_{s=0}
        }
    }
    \label{eq:resonant-alpha-general}
\end{equation}
	 
This coefficient determines the full resonant solution. For $n<N$, the coefficients are analytic near $P_*$, so
\begin{equation}
    h_n\bigl(\gamma(s)\bigr)\longrightarrow h_n^*,
    \qquad n<N.
\end{equation}
Moreover, the affine relation from Section~\ref{subsec:dictionary-higher-coefficients},
\begin{equation}
    h_{N+j}=P_{N+j}+Q_{N+j}h_N,
    \qquad j\geq1,
\end{equation}
has $P_{N+j}$ and $Q_{N+j}$ analytic near resonance. Hence, once $h_N$ has the finite limit in Eq.~\eqref{eq:resonant-alpha-general}, all higher coefficients have unique finite limits as well.
	 
The limiting coefficients satisfy the resonant recurrence and have $n=N$ coefficient $\alpha_N^{(\gamma)}$. Therefore Eq.~\eqref{eq:resonant-general-solution} gives
\begin{equation}
    \boxed{
        R_{\rm sel}^{(\gamma)}
        =
        H_0^{(N)}
        +
        \alpha_N^{(\gamma)}H_{\rm out}^{(N)}
    }
    \label{eq:resonant-selected-solution}
\end{equation}
	 
Applying the same local majorant argument used in Appendix~\ref{app:frobenius-convergence}, restricted to the curve $\gamma$, shows that this coefficientwise limit is the limit of actual local solutions in a sufficiently small common neighborhood of the horizon.
	 
We summarize the result as follows.
	 
\begin{theorem}[One-parameter resonant limit]
    \label{thm:resonant-path-limit}
    Let $P_*$ be a resonant point satisfying
    \begin{equation}
        D_N(P_*)=0,
        \qquad
        \Omega_N(P_*)=0.
    \end{equation}
    Let $\gamma(s)$ be a $C^1$ curve approaching $P_*$ and suppose that Eq.~\eqref{eq:resonant-transverse-D} holds. Then the resonant limit of the nonresonant ingoing solution normalized by $h_0=1$ is
    \begin{equation}
        R_{\rm sel}^{(\gamma)}
        =
        H_0^{(N)}
        +
        \alpha_N^{(\gamma)}H_{\rm out}^{(N)},
    \end{equation}
    where
    \begin{equation}
        \alpha_N^{(\gamma)}
        =
        -
        \frac{
            \left.\dfrac{d}{ds}\Omega_N(\gamma(s))\right|_{s=0}
        }{
            \left.\dfrac{d}{ds}D_N(\gamma(s))\right|_{s=0}
        }.
        \label{eq:resonant-path-theorem}
    \end{equation}
\end{theorem}
	 
Equation~\eqref{eq:resonant-path-theorem} is invariant under a nonzero analytic rescaling of the recurrence normalization. Indeed, if the entire $N$th recurrence is multiplied by an analytic function $f(P)\neq0$, then
\begin{equation}
    D_N\longrightarrow fD_N,
    \qquad
    \Omega_N\longrightarrow f\Omega_N.
\end{equation}
Since $D_N=\Omega_N=0$ at $P_*$,
\begin{align}
    d(fD_N)|_{P_*}&=f(P_*)\,dD_N|_{P_*},\\
    d(f\Omega_N)|_{P_*}&=f(P_*)\,d\Omega_N|_{P_*},
\end{align}
so the ratio of the two derivatives is unchanged.
	 
If instead
\begin{equation}
    \left.\frac{d}{ds}D_N(\gamma(s))\right|_{s=0}=0,
\end{equation}
then Eq.~\eqref{eq:resonant-path-theorem} does not apply. In that case the higher-order terms of $D_N\circ\gamma$ and $\Omega_N\circ\gamma$ must be compared, and the selected resonant coefficient can depend on second- or higher-order jets.

\subsection{Frequency continuation at fixed background}
\label{subsec:resonant-frequency}
	 
Now fix the background so that the surface gravity $\kappa$ remains constant under the parameter variation. Taking $\omega$ as a local parameter along the curve gives
\begin{equation}
    \nu(\omega)=\frac{i\omega}{\kappa},
\end{equation}
and hence
\begin{equation}
    D_N(\omega)
    =
    N\left(N-\frac{i\omega}{\kappa}\right).
    \label{eq:resonant-DN-omega}
\end{equation}
The resonant frequency is
\begin{equation}
    \omega_N=-iN\kappa,
\end{equation}
and
\begin{equation}
    \left.\frac{dD_N}{d\omega}\right|_{\omega_N}
    =
    -\frac{iN}{\kappa}\neq0.
    \label{eq:resonant-DN-omega-derivative}
\end{equation}
	 
Therefore Theorem~\ref{thm:resonant-path-limit} reduces to
\begin{equation}
    \boxed{
        \alpha_N^{(\gamma)}
        =
        \frac{\kappa}{iN}
        \left.\frac{d}{d\omega}\Omega_N\bigl(\gamma(\omega)\bigr)\right|_{\omega=\omega_N}
    }
    \label{eq:resonant-alpha-omega}
\end{equation}
	 
The derivative here is generally a total derivative along the curve $\gamma$, rather than the partial derivative $\partial_\omega\Omega_N$. For example, if
\begin{equation}
    \gamma(\omega)
    =
    \bigl(\omega,k(\omega),\Delta(\omega),\mathbf g(\omega)\bigr),
\end{equation}
then
\begin{align}
    \left.\frac{d}{d\omega}\Omega_N\bigl(\gamma(\omega)\bigr)\right|_{\omega_N}
    =\Bigg[&
    \partial_\omega\Omega_N
    +k'(\omega_N)\partial_k\Omega_N
    +\Delta'(\omega_N)\partial_\Delta\Omega_N
    \nonumber\\
    &+\sum_A g_A'(\omega_N)\partial_{g_A}\Omega_N
    \Bigg]_{P_*}.
    \label{eq:resonant-total-derivative}
\end{align}
	 
In particular, if all nonfrequency parameters are held fixed,
\begin{equation}
    k=k_*,
    \qquad
    \Delta=\Delta_*,
    \qquad
    \mathbf g=\mathbf g_*,
\end{equation}
then
\begin{equation}
    \boxed{
        \alpha_N^{\rm fixed}
        =
        \frac{\kappa}{iN}
        \left.\partial_\omega\Omega_N\right|_{P_*}
    }
    \label{eq:resonant-alpha-fixed}
\end{equation}
	 
It is important to distinguish uniqueness of the fixed-parameter continuation from the condition $\alpha_N^{\rm fixed}=0$. Fixing $(k,\Delta,\mathbf g)$ determines a unique trajectory in parameter space, so Eq.~\eqref{eq:resonant-alpha-fixed} gives a unique value. That value, however, need not vanish in general.
	 
Thus
\begin{equation}
    \text{fixed-parameter continuation}
    \quad\Longrightarrow\quad
    \alpha_N^{\rm fixed}\text{ is uniquely determined},
\end{equation}
whereas
\begin{equation}
    \alpha_N^{\rm fixed}=0
\end{equation}
is an additional condition.

\subsection{An additional selection condition at exact resonance}
\label{subsec:resonant-ingoing-representative}
	 
In Eq.~\eqref{eq:resonant-H0}, $H_0^{(N)}$ is the resonant solution for which the coefficient at $n=N$ is chosen to vanish. Thus
\begin{equation}
    H_0^{(N)}
    =
    1+h_1^*x+\cdots+h_{N-1}^*x^{N-1}
    +0\cdot x^N+O(x^{N+1}),
\end{equation}
and in this normalization it contains no $H_{\rm out}^{(N)}$ component.
	 
On the other hand, taking the nonresonant ingoing solution to exact resonance along a curve $\gamma$ selects, as shown in Section~\ref{subsec:resonant-general-curve},
\begin{equation}
    R_{\rm sel}^{(\gamma)}
    =
    H_0^{(N)}
    +
    \alpha_N^{(\gamma)}H_{\rm out}^{(N)}.
    \label{eq:resonant-selected-repeat}
\end{equation}
An important distinction is required here. At exact resonance,
\begin{equation}
    H_{\rm out}^{(N)}=x^N\bigl(1+O(x)\bigr)
\end{equation}
is analytic at the horizon, so horizon smoothness alone cannot exclude a component proportional to $H_{\rm out}^{(N)}$. However, $H_{\rm out}^{(N)}$ is obtained by taking the second nonresonant Frobenius solution
\begin{equation}
    H_{\rm out}(x;P)=x^{\nu(P)}\bigl(1+O(x)\bigr)
\end{equation}
to resonance. If, in addition to smoothness, one imposes the physical requirement that the resonant representative contain no component along the solution obtained by taking the second nonresonant Frobenius branch to resonance, then one must have
\begin{equation}
    R_{\rm sel}^{(\gamma)}=H_0^{(N)}.
\end{equation}
Importantly, this condition does not follow from $\Omega_N(P_*)=0$ or from $C^\infty$ horizon smoothness. The condition $\Omega_N(P_*)=0$ guarantees the existence of a two-dimensional smooth resonant solution space, whereas excluding the component associated with the second nonresonant Frobenius branch is an additional physical requirement that selects one particular resonant solution within that space.

Therefore the limit along a parameter-space curve $\gamma$ agrees exactly with $H_0^{(N)}$ if and only if
\begin{equation}
    \boxed{\alpha_N^{(\gamma)}=0}
    \label{eq:resonant-alpha-zero}
\end{equation}

Under the nondegeneracy condition of Theorem~\ref{thm:resonant-path-limit},
\begin{equation}
    \left.\frac{d}{ds}D_N\bigl(\gamma(s)\bigr)\right|_{s=0}\neq0,
\end{equation}
we have
\begin{equation}
    \alpha_N^{(\gamma)}
    =
    -\frac{
        \left.\dfrac{d}{ds}\Omega_N\bigl(\gamma(s)\bigr)\right|_{s=0}
    }{
        \left.\dfrac{d}{ds}D_N\bigl(\gamma(s)\bigr)\right|_{s=0}
    }.
\end{equation}
Since the denominator is nonzero,
\begin{equation}
    \boxed{
        \alpha_N^{(\gamma)}=0
        \quad\Longleftrightarrow\quad
        \left.\frac{d}{ds}\Omega_N\bigl(\gamma(s)\bigr)\right|_{s=0}=0
    }
    \label{eq:resonant-preservation-condition}
\end{equation}
In solution-space terms, $\alpha_N^{(\gamma)}=0$ removes the component coming from the second nonresonant Frobenius solution. In parameter-space terms,
\begin{equation}
    \left.d\Omega_N\right|_{P_*}(\dot\gamma_*)=0
\end{equation}
is the corresponding first-order condition on the approach curve. Thus a general limit along a parameter-space curve and this additional physical condition are not the same in general: the former may select
\begin{equation}
    H_0^{(N)}+\alpha_N^{(\gamma)}H_{\rm out}^{(N)},
\end{equation}
whereas the latter additionally requires
\begin{equation}
    \boxed{\alpha_N^{(\gamma)}=0}.
\end{equation}
This shows that the condition
\begin{equation}
    \Omega_N(P_*)=0
\end{equation}
alone does not determine which nonzero resonant solution is selected as the limit of the nonresonant ingoing family. It guarantees two independent smooth horizon solutions, while the coefficient $\alpha_N^{(\gamma)}$ selected within that two-dimensional space is controlled by the first-order variation of $\Omega_N$ along the approach curve.
	 
Equivalently, Eq.~\eqref{eq:resonant-preservation-condition} may be written in differential notation on parameter space as
\begin{equation}
    \left.\frac{d}{ds}\Omega_N\bigl(\gamma(s)\bigr)\right|_{s=0}
    =
    \left.d\Omega_N\right|_{P_*}(\dot\gamma_*),
    \qquad
    \dot\gamma_*
    :=
    \left.\frac{d\gamma}{ds}\right|_{s=0}.
    \label{eq:resonant-differential}
\end{equation}
Here $\dot\gamma_*$ is simply the first-order direction of the curve approaching $P_*$ in parameter space. Hence Eq.~\eqref{eq:resonant-preservation-condition} can also be expressed as
\begin{equation}
    \alpha_N^{(\gamma)}=0
    \quad\Longleftrightarrow\quad
    \left.d\Omega_N\right|_{P_*}(\dot\gamma_*)=0.
    \label{eq:resonant-tangent-condition}
\end{equation}
	 
The point is that this condition does not require
\begin{equation}
    \Omega_N\bigl(\gamma(s)\bigr)\equiv0
\end{equation}
along the entire curve. It requires only the first-order condition at $P_*$,
\begin{equation}
    \left.\frac{d}{ds}\Omega_N\bigl(\gamma(s)\bigr)\right|_{s=0}=0.
\end{equation}
Thus $\Omega_N$ must be stationary to first order along the chosen curve as it passes through $P_*$. Nonzero higher-order variations of $\Omega_N$ are fully compatible with Eq.~\eqref{eq:resonant-preservation-condition}.

	 \subsection{Moving momentum channel}
	 \label{subsec:resonant-moving-k}
	 
A useful special case is obtained by fixing $\Delta$ and the background while allowing the momentum to vary with the frequency:
\begin{equation}
    \Delta=\Delta_*,
    \qquad
    \mathbf g=\mathbf g_*,
    \qquad
    k=k(\omega),
    \qquad
    k(\omega_N)=k_*.
    \label{eq:resonant-moving-k-path}
\end{equation}
	 
Then
\begin{equation}
    \left.\frac{d}{d\omega}\Omega_N(\gamma(\omega))\right|_{\omega_N}
    =
    \left[\partial_\omega\Omega_N+k'(\omega_N)\partial_k\Omega_N\right]_{P_*},
    \label{eq:resonant-moving-k-total}
\end{equation}
and therefore
\begin{equation}
    \boxed{
        \alpha_N^{(k)}
        =
        \frac{\kappa}{iN}
        \left[\partial_\omega\Omega_N+k'(\omega_N)\partial_k\Omega_N\right]_{P_*}
    }
    \label{eq:resonant-alpha-moving-k}
\end{equation}
	 
Requiring $\alpha_N^{(k)}=0$, so that the limit agrees with the resonant solution $H_0^{(N)}$, gives
\begin{equation}
    \left[\partial_\omega\Omega_N+k'(\omega_N)\partial_k\Omega_N\right]_{P_*}=0.
    \label{eq:resonant-moving-k-condition}
\end{equation}
	 
In particular, if
\begin{equation}
    \left.\partial_k\Omega_N\right|_{P_*}\neq0,
    \label{eq:resonant-k-nondegenerate}
\end{equation}
then
\begin{equation}
    \boxed{
        k'(\omega_N)
        =
        -\frac{(\partial_\omega\Omega_N)_{P_*}}{(\partial_k\Omega_N)_{P_*}}
    }
    \label{eq:resonant-k-slope}
\end{equation}
	 
Thus, in this case, the slope is not arbitrary: requiring the limit to select $H_0^{(N)}$ fixes the first-order momentum--frequency slope through the pole-skipping point.
	 
Several degenerate cases should be distinguished.
	 
First, if
\begin{equation}
    (\partial_k\Omega_N)_{P_*}=0,
    \qquad
    (\partial_\omega\Omega_N)_{P_*}\neq0,
\end{equation}
then no finite value of $k'(\omega_N)$ can satisfy Eq.~\eqref{eq:resonant-moving-k-condition}.
	 
Second, if
\begin{equation}
    (\partial_k\Omega_N)_{P_*}
    =
    (\partial_\omega\Omega_N)_{P_*}
    =0,
\end{equation}
then $k'(\omega_N)$ is not determined at first order. One must then examine the second-order expansion or higher-order jets.
	 
Third, in an isotropic radial equation the momentum dependence is often analytic in
\begin{equation}
    q:=k^2
\end{equation}
rather than in $k$ itself. In particular, at $k_*=0$,
\begin{equation}
    \partial_k\Omega_N=2k\,\partial_q\Omega_N,
\end{equation}
so $(\partial_k\Omega_N)_{P_*}=0$ may hold automatically. In this case it is more natural to use $q=k^2$ as the basic spectral variable and write
\begin{equation}
    \boxed{
        q'(\omega_N)
        =
        -\frac{(\partial_\omega\Omega_N)_{P_*}}{(\partial_q\Omega_N)_{P_*}}
    }
    \label{eq:resonant-q-slope}
\end{equation}
	 
If the direction of the full vector momentum is also allowed to vary independently for $d>1$, the single scalar condition
\begin{equation}
    d\Omega_N|_{P_*}(\dot\gamma_*)=0
\end{equation}
cannot determine all components of $\mathbf k'$. Equation~\eqref{eq:resonant-k-slope} therefore applies directly to a scalar momentum channel with fixed momentum direction, or to cases in which isotropy leaves only $k^2$ as the independent spectral variable.

	\subsection{Limits along parameter-space curves and projective classes of boundary coefficients}
	\label{subsec:resonant-boundary-limit}
	
Finally, we distinguish explicitly the simultaneous vanishing of the regularized boundary connection coefficients found in Section~\ref{sec:dictionary} from the nonzero resonant solution obtained here along a curve in parameter space, and we describe how the local resonant limit is transmitted to the boundary coefficients.
	
Continue the two independent smooth horizon solutions $H_0^{(N)}$ and $H_{\rm out}^{(N)}$ at exact resonance to the boundary and write
\begin{align}
    H_0^{(N)}&=A_0R_1+B_0R_2,
    \label{eq:resonant-H0-boundary}
    \\
    H_{\rm out}^{(N)}&=A_{\rm out}R_1+B_{\rm out}R_2.
    \label{eq:resonant-Hout-boundary}
\end{align}
Here $R_1$ and $R_2$ are the independent boundary-normalized solutions defined in Section~\ref{subsec:setup-boundary-basis}.
	
Since $H_0^{(N)}$ and $H_{\rm out}^{(N)}$ are also independent, the two boundary coefficient vectors
\begin{equation}
    (A_0,B_0),
    \qquad
    (A_{\rm out},B_{\rm out})
\end{equation}
are not proportional.
	
As shown above, taking the nonresonant ingoing solution normalized by $h_0=1$ to exact resonance along a parameter-space curve $\gamma$ selects the nonzero resonant solution
\begin{equation}
    R_{\rm sel}^{(\gamma)}
    =
    H_0^{(N)}
    +
    \alpha_N^{(\gamma)}H_{\rm out}^{(N)}.
    \label{eq:resonant-selected-solution-boundary}
\end{equation}
With the present Frobenius normalization, $\alpha_N^{(\gamma)}$ is the coefficient at order $n=N$ selected at resonance by the curve $\gamma$.
	
Substituting Eqs.~\eqref{eq:resonant-H0-boundary} and \eqref{eq:resonant-Hout-boundary} into Eq.~\eqref{eq:resonant-selected-solution-boundary} gives
\begin{align}
    R_{\rm sel}^{(\gamma)}
    &=
    \left(A_0+\alpha_N^{(\gamma)}A_{\rm out}\right)R_1
    \nonumber\\
    &\quad+
    \left(B_0+\alpha_N^{(\gamma)}B_{\rm out}\right)R_2.
    \label{eq:resonant-selected-boundary}
\end{align}
We therefore define the boundary coefficients of the resonant solution selected by $\gamma$ as
\begin{align}
    A_\gamma&:=A_0+\alpha_N^{(\gamma)}A_{\rm out},
    \label{eq:resonant-Agamma}
    \\
    B_\gamma&:=B_0+\alpha_N^{(\gamma)}B_{\rm out}.
    \label{eq:resonant-Bgamma}
\end{align}
	
These coefficients must not be confused with the regularized coefficients $\widehat A,\widehat B$ of Section~\ref{sec:dictionary}. To display the relation between the two pairs, write at nonresonant points $P=\gamma(s)$, $s\neq0$,
\begin{equation}
    H_{\rm in}(P)=A(P)R_1+B(P)R_2.
\end{equation}
Multiplication by the factor $g_N$ that removes the simple parameter pole gives
\begin{equation}
    H_{\rm in}^{\rm reg}(P)=g_N(P)H_{\rm in}(P)
\end{equation}
and hence
\begin{equation}
    H_{\rm in}^{\rm reg}(P)
    =
    \widehat A(P)R_1+\widehat B(P)R_2,
\end{equation}
where
\begin{equation}
    \boxed{
        \widehat A(P)=g_N(P)A(P),
        \qquad
        \widehat B(P)=g_N(P)B(P)
    }
    \label{eq:resonant-hat-vs-ordinary}
\end{equation}
	
At the pole-skipping point $P_*$,
\begin{equation}
    g_N(P_*)=0,
    \qquad
    \Omega_N(P_*)=0,
\end{equation}
and the identity at exact resonance established in Section~\ref{sec:dictionary} gives
\begin{equation}
    H_{\rm in}^{\rm reg}(P_*)=0.
\end{equation}
Therefore
\begin{equation}
    \boxed{\widehat A(P_*)=\widehat B(P_*)=0.}
    \label{eq:resonant-hat-zero}
\end{equation}
This equation states that the regularized family has the zero solution as its value at $P_*$.
	
By contrast, $R_{\rm sel}^{(\gamma)}$ is the limit along the parameter-space curve $\gamma$ of the nonresonant ingoing family with the normalization $h_0=1$, and hence
\begin{equation}
    R_{\rm sel}^{(\gamma)}\neq0.
\end{equation}
Since $R_1$ and $R_2$ are independent,
\begin{equation}
    \boxed{(A_\gamma,B_\gamma)\neq(0,0).}
    \label{eq:resonant-Agamma-Bgamma-nonzero}
\end{equation}
Thus the boundary coefficient pair of the resonant solution obtained along the curve defines a well-defined projective class
\begin{equation}
    \boxed{[A_\gamma:B_\gamma]\in\mathbb{CP}^1.}
    \label{eq:resonant-projective-limit}
\end{equation}
	
The important point is that
\begin{equation}
    (\widehat A(P_*),\widehat B(P_*))=(0,0)
\end{equation}
does not define a projective point; in particular, $[0:0]$ is not an element of $\mathbb{CP}^1$. At the exact point, the regularized coefficient pair by itself carries no information about a boundary projective class. That class is obtained only after a curve approaching $P_*$ is specified and the corresponding limit is taken.
	
This distinction is especially transparent at the level of the coefficients. Suppose that along a curve $\gamma$,
\begin{equation}
    A(\gamma(s))\longrightarrow A_\gamma,
    \qquad
    B(\gamma(s))\longrightarrow B_\gamma.
\end{equation}
At the same time,
\begin{equation}
    g_N(\gamma(s))\longrightarrow0.
\end{equation}
Equation~\eqref{eq:resonant-hat-vs-ordinary} then implies
\begin{align}
    \widehat A(\gamma(s))
    &=g_N(\gamma(s))A(\gamma(s))\longrightarrow0,\\
    \widehat B(\gamma(s))
    &=g_N(\gamma(s))B(\gamma(s))\longrightarrow0.
\end{align}
Thus it is possible simultaneously to have
\begin{equation}
    (\widehat A,\widehat B)\longrightarrow(0,0)
\end{equation}
and
\begin{equation}
    (A,B)\longrightarrow(A_\gamma,B_\gamma)\neq(0,0).
\end{equation}
There is no contradiction between these two statements.
	
Moreover, for $s\neq0$ one has $g_N(\gamma(s))\neq0$, so
\begin{equation}
    \frac{\widehat B(\gamma(s))}{\widehat A(\gamma(s))}
    =
    \frac{B(\gamma(s))}{A(\gamma(s))}.
    \label{eq:resonant-ratio-equality-off}
\end{equation}
Hence, if $A_\gamma\neq0$,
\begin{equation}
    \boxed{
        \lim_{s\to0}
        \frac{\widehat B(\gamma(s))}{\widehat A(\gamma(s))}
        =
        \frac{B_\gamma}{A_\gamma}
        =
        \frac{B_0+\alpha_N^{(\gamma)}B_{\rm out}}{A_0+\alpha_N^{(\gamma)}A_{\rm out}}
    }
    \label{eq:resonant-boundary-ratio}
\end{equation}
	
Thus, at the exact pole-skipping point, $\widehat A$ and $\widehat B$ vanish simultaneously and their ratio is not defined pointwise. Once an approach curve $\gamma$ is specified, however, the limit of the $0/0$ ratio generally has the definite value $B_\gamma/A_\gamma$.
	
If instead
\begin{equation}
    A_\gamma=0,
    \qquad
    B_\gamma\neq0,
\end{equation}
then the projective class obtained along $\gamma$ is
\begin{equation}
    [A_\gamma:B_\gamma]=[0:1],
\end{equation}
corresponding to the pole case $B/A=\infty$. Thus finite ratios and the pole case are naturally unified by regarding the projective class of the boundary coefficient pair along $\gamma$ as a point of $\mathbb{CP}^1$.
	
The dependence on the approach curve is now explicit. If two curves $\gamma_1$ and $\gamma_2$ select
\begin{equation}
    \alpha_N^{(\gamma_1)}\neq\alpha_N^{(\gamma_2)},
\end{equation}
then generically
\begin{equation}
    [A_{\gamma_1}:B_{\gamma_1}]
    \neq
    [A_{\gamma_2}:B_{\gamma_2}].
\end{equation}
Thus even at the same exact pole-skipping point, different curves can select different nonzero resonant bulk solutions and different projective classes $[A_{\gamma_i}:B_{\gamma_i}]$.
	
The results of this section can be summarized in three stages. First,
\begin{equation}
    \boxed{D_N(P_*)=0}
\end{equation}
fixes the integer resonance. Second,
\begin{equation}
    \boxed{\Omega_N(P_*)=0}
\end{equation}
ensures the existence of two independent smooth horizon directions at exact resonance. Third, under the first-order nondegeneracy condition,
\begin{equation}
    \boxed{
        \alpha_N^{(\gamma)}
        =
        -\frac{
            \left.\dfrac{d}{ds}\Omega_N(\gamma(s))\right|_{s=0}
        }{
            \left.\dfrac{d}{ds}D_N(\gamma(s))\right|_{s=0}
        }
    }
    \label{eq:resonant-summary-alpha}
\end{equation}
determines the actual resonant solution selected by the curve $\gamma$ within that two-dimensional smooth solution space.
	
Consequently, the simultaneous zero
\begin{equation}
    \widehat A(P_*)=\widehat B(P_*)=0
\end{equation}
from Section~\ref{sec:dictionary} and the nonzero projective pair
\begin{equation}
    [A_\gamma:B_\gamma]\in\mathbb{CP}^1
\end{equation}
obtained here are not contradictory statements.

	 \section{Pole-skipping and holographic Green functions}
	 \label{sec:pole-skipping-green}
	 
In the preceding sections we analyzed separately the resonant Frobenius problem at the horizon and the connection problem at the boundary. Section~\ref{sec:dictionary} showed that at exact resonance
\begin{equation}
    \Omega_N(P_*)=0
    \quad\Longleftrightarrow\quad
    \widehat A(P_*)=\widehat B(P_*)=0,
    \label{eq:green-hat-simultaneous-zero}
\end{equation}
whereas Section~\ref{sec:resonant-limit} showed that, once a parameter-space curve $\gamma$ is specified, the nonresonant ingoing family normalized by $h_0=1$ selects the nonzero resonant solution
\begin{equation}
    R_{\rm sel}^{(\gamma)}
    =
    H_0^{(N)}
    +
    \alpha_N^{(\gamma)}H_{\rm out}^{(N)}.
    \label{eq:green-selected-resonant}
\end{equation}
	 
In this section we connect these two results directly to the pole-skipping structure of the holographic retarded Green function. In particular, we distinguish the regularized $0/0$ at the exact point from the finite or infinite Green-function limit obtained along a specified approach curve.

\subsection{Retarded Green function and boundary coefficient ratio}
\label{subsec:green-retarded-ratio}
	 
At a nonresonant parameter value, continue the ingoing solution normalized by $h_0=1$ to the boundary and write
\begin{equation}
    H_{\rm in}=AR_1+BR_2.
    \label{eq:green-ordinary-decomposition}
\end{equation}
According to the convention fixed in Section~\ref{subsec:setup-boundary-basis}, $R_1$ is the source-normalized solution and $R_2$ is the response-normalized solution.
	 
In the standard holographic prescription, up to an overall normalization and a convention-dependent prefactor, the retarded Green function is
\begin{equation}
    \boxed{G_R\propto\frac{B}{A}}
    \label{eq:green-retarded-ratio}
\end{equation}
\cite{SonStarinets2002,HerzogSon2003,SkenderisVanRees2008}.
	 
Thus $A=0$ with $B\neq0$ gives a pole, whereas $B=0$ with $A\neq0$ gives a zero. At a generic nonresonant point, the ingoing horizon condition selects a unique solution up to overall normalization, and therefore the ratio $B/A$ is uniquely determined.
	 
Near resonance, Section~\ref{sec:dictionary} introduced the factor $g_N$ that removes the simple parameter pole and defined
\begin{equation}
    H_{\rm in}^{\rm reg}=g_NH_{\rm in}.
\end{equation}
Hence, if
\begin{equation}
    H_{\rm in}^{\rm reg}=\widehat A R_1+\widehat B R_2,
\end{equation}
then
\begin{equation}
    \widehat A=g_NA,
    \qquad
    \widehat B=g_NB.
    \label{eq:green-hat-ordinary-relation}
\end{equation}
	 
Away from resonance $g_N\neq0$, so
\begin{equation}
    \boxed{\frac{B}{A}=\frac{\widehat B}{\widehat A}}
    \label{eq:green-hat-ratio-off-resonance}
\end{equation}
Thus multiplication by the pole-removing factor does not change the nonresonant retarded Green-function ratio.

\subsection{Regularized \texorpdfstring{$0/0$}{0/0} at the exact pole-skipping point}
\label{subsec:green-zero-zero}
	 
Let
\begin{equation}
    P_*=(\omega_N,\lambda_*)
\end{equation}
be a resonant point satisfying
\begin{equation}
    D_N(P_*)=0,
    \qquad
    \Omega_N(P_*)=0.
    \label{eq:green-pole-skipping-local}
\end{equation}
	 
The first condition is the integer Frobenius resonance condition, while the second is the solvability condition at the resonant order. By the horizon-to-boundary result of Section~\ref{sec:dictionary},
\begin{equation}
    \Omega_N(P_*)=0
    \quad\Longleftrightarrow\quad
    \widehat A(P_*)=\widehat B(P_*)=0.
    \label{eq:green-pole-skipping-boundary}
\end{equation}
	 
Therefore, at the exact point the regularized source--response ratio formally takes the form
\begin{equation}
    \frac{\widehat B(P_*)}{\widehat A(P_*)}=\frac{0}{0}.
    \label{eq:green-zero-zero-form}
\end{equation}
	 
It is essential that Eq.~\eqref{eq:green-zero-zero-form} not be interpreted as saying that a single nonzero resonant solution has both source and response coefficients equal to zero. In fact,
\begin{equation}
    H_{\rm in}^{\rm reg}(P_*)=0,
\end{equation}
so $\widehat A(P_*)=\widehat B(P_*)=0$ are the boundary coefficients of the zero solution.
	 
Consequently, the expression
\begin{equation}
    [\widehat A(P_*):\widehat B(P_*)]
\end{equation}
does not define a projective class at the exact point. In particular, $[0:0]$ is not an element of $\mathbb{CP}^1$.
	 
The relation between Green-function nonuniqueness and degeneracy of the near-horizon ingoing solution has been established in previous holographic analyses \cite{BlakeDavisonVegh2020,NatsuumeOkamura2019,NatsuumeOkamura2020}. In the present formulation, the $0/0$ must be understood together with the distinction involving the regularized family above. The regularized coefficient values at the exact point do not by themselves select a boundary source--response ratio; to determine such a ratio one must additionally specify how $P_*$ is approached.

	 \subsection{Green-function limits along curves in parameter space}
	 \label{subsec:green-pathwise-limit}
	 
Fix a curve in parameter space,
\begin{equation}
    \gamma(s),
    \qquad
    \gamma(0)=P_*.
\end{equation}
As shown in Section~\ref{sec:resonant-limit}, under the first-order nondegeneracy condition
\begin{equation}
    \left.\frac{d}{ds}D_N(\gamma(s))\right|_{s=0}\neq0,
    \label{eq:green-path-nondegenerate-D}
\end{equation}
the curve selects
\begin{equation}
    \alpha_N^{(\gamma)}
    =
    -\frac{
        \left.\dfrac{d}{ds}\Omega_N(\gamma(s))\right|_{s=0}
    }{
        \left.\dfrac{d}{ds}D_N(\gamma(s))\right|_{s=0}
    }.
    \label{eq:green-alpha-path}
\end{equation}
	 
The nonzero solution selected at exact resonance is therefore
\begin{equation}
    R_{\rm sel}^{(\gamma)}
    =
    H_0^{(N)}
    +
    \alpha_N^{(\gamma)}H_{\rm out}^{(N)}.
\end{equation}
	 
Continuing the two horizon solutions at resonance to the boundary, write
\begin{align}
    H_0^{(N)}&=A_0R_1+B_0R_2,\\
    H_{\rm out}^{(N)}&=A_{\rm out}R_1+B_{\rm out}R_2.
\end{align}
Then
\begin{equation}
    R_{\rm sel}^{(\gamma)}=A_\gamma R_1+B_\gamma R_2,
\end{equation}
where
\begin{align}
    A_\gamma&=A_0+\alpha_N^{(\gamma)}A_{\rm out},
    \label{eq:green-Agamma}\\
    B_\gamma&=B_0+\alpha_N^{(\gamma)}B_{\rm out}.
    \label{eq:green-Bgamma}
\end{align}
	 
Since $R_{\rm sel}^{(\gamma)}$ is nonzero and $R_1,R_2$ are independent,
\begin{equation}
    (A_\gamma,B_\gamma)\neq(0,0).
\end{equation}
Hence
\begin{equation}
    [A_\gamma:B_\gamma]\in\mathbb{CP}^1
\end{equation}
is well defined.
	 
For $s\neq0$, Eq.~\eqref{eq:green-hat-ratio-off-resonance} holds. Therefore, when $A_\gamma\neq0$,
\begin{equation}
    \boxed{
        \lim_{s\to0}G_R(\gamma(s))
        \propto
        \lim_{s\to0}
        \frac{\widehat B(\gamma(s))}{\widehat A(\gamma(s))}
        =
        \frac{B_\gamma}{A_\gamma}
    }
    \label{eq:green-pathwise-limit}
\end{equation}
	 
Thus, although the pointwise ratio at the exact point is undefined because
\begin{equation}
    \frac{\widehat B}{\widehat A}=\frac00,
\end{equation}
choosing a curve $\gamma$ generally determines a well-defined projective class $[A_\gamma:B_\gamma]$ for the limiting boundary coefficient pair.
	 
Using Eqs.~\eqref{eq:green-Agamma} and \eqref{eq:green-Bgamma}, this limit can be written as
\begin{equation}
    \boxed{
        G_R^{(\gamma)}
        \propto
        \frac{B_0+\alpha_N^{(\gamma)}B_{\rm out}}{A_0+\alpha_N^{(\gamma)}A_{\rm out}}
    }
    \label{eq:green-alpha-ratio}
\end{equation}
Here $G_R^{(\gamma)}$ denotes the Green-function limit obtained by taking the nonresonant retarded solution to the pole-skipping point along the parameter-space curve $\gamma$. As Eq.~\eqref{eq:green-alpha-path} shows, for unrestricted approach curves $\alpha_N^{(\gamma)}$ can be nonzero, and different curves can select different Green-function limits.

If one imposes instead the additional physical condition described in Section~\ref{subsec:resonant-ingoing-representative}, namely that the resonant solution contain no component coming from the second nonresonant Frobenius solution, then
\begin{equation}
    \alpha_N^{(\gamma)}=0.
\end{equation}
The selected resonant solution is then $H_0^{(N)}$, and its boundary coefficient pair defines the projective class
\begin{equation}
    [A_0:B_0]\in\mathbb{CP}^1.
\end{equation}
If $A_0\neq0$, the retarded Green-function ratio under this additional condition is
\begin{equation}
    \boxed{
        G_R^{(\alpha=0)}\propto\frac{B_0}{A_0}
    }
    \label{eq:green-branch-preserving}
\end{equation}
whereas if $A_0=0$ and $B_0\neq0$, the corresponding projective class is $[0:1]$. Thus the ambiguity associated with arbitrary curves is not in conflict with the additional physical condition $\alpha_N^{(\gamma)}=0$: the former describes the general limit along a parameter-space curve, while the latter selects one particular resonant solution among the allowed possibilities.

\subsection{Pole, zero, and finite limits along parameter-space curves}
\label{subsec:green-pole-zero-directions}
	 
Equation~\eqref{eq:green-alpha-ratio} shows directly that different curves approaching the pole-skipping point can select qualitatively different boundary limits.
	 
If
\begin{equation}
    A_\gamma=0,
    \qquad
    B_\gamma\neq0,
\end{equation}
then
\begin{equation}
    [A_\gamma:B_\gamma]=[0:1]
\end{equation}
and
\begin{equation}
    G_R^{(\gamma)}=\infty,
\end{equation}
corresponding to a pole.
	 
Conversely, if
\begin{equation}
    B_\gamma=0,
    \qquad
    A_\gamma\neq0,
\end{equation}
then
\begin{equation}
    [A_\gamma:B_\gamma]=[1:0]
\end{equation}
and
\begin{equation}
    G_R^{(\gamma)}=0,
\end{equation}
corresponding to a zero.
	 
In all other cases with
\begin{equation}
    A_\gamma B_\gamma\neq0,
\end{equation}
one obtains a finite nonzero limit,
\begin{equation}
    0<\left|\frac{B_\gamma}{A_\gamma}\right|<\infty.
\end{equation}
	 
It is therefore misleading to characterize a pole-skipping point as carrying only a pole or only a zero. At the exact point the regularized source and response coefficients vanish simultaneously, so their ratio is undetermined. A pole, a zero, or a finite value is selected only after an approach direction in parameter space is specified.
	  
\subsection{Different values of \texorpdfstring{$\alpha_N$}{alpha N} select different boundary projective classes}
\label{subsec:green-alpha-injectivity}
	
We now make explicit how the dependence on the approach curve is transmitted to the boundary coefficients. At exact resonance, the solution selected by a curve $\gamma$ is
\begin{equation}
    R_{\rm sel}^{(\gamma)}
    =
    H_0^{(N)}
    +
    \alpha_N^{(\gamma)}H_{\rm out}^{(N)},
\end{equation}
and all such solutions are normalized by the same condition $h_0=1$ at the horizon.
	
After continuation to the boundary,
\begin{equation}
    R_{\rm sel}^{(\gamma)}=A_\gamma R_1+B_\gamma R_2,
\end{equation}
where
\begin{align}
    A_\gamma&=A_0+\alpha_N^{(\gamma)}A_{\rm out},\\
    B_\gamma&=B_0+\alpha_N^{(\gamma)}B_{\rm out}.
\end{align}
	
Suppose two different curves $\gamma_1$ and $\gamma_2$ select
\begin{equation}
    \alpha_N^{(\gamma_1)}\neq\alpha_N^{(\gamma_2)}.
\end{equation}
If they gave the same projective class of the boundary coefficient pair, then there would exist a constant $c\neq0$ such that
\begin{equation}
    (A_{\gamma_1},B_{\gamma_1})
    =
    c\,(A_{\gamma_2},B_{\gamma_2}).
\end{equation}
	
Because $R_1$ and $R_2$ are independent, this would imply for the global solutions that
\begin{equation}
    R_{\rm sel}^{(\gamma_1)}
    =
    c\,R_{\rm sel}^{(\gamma_2)}.
    \label{eq:green-proportional-selected}
\end{equation}
	
Both solutions, however, are normalized by $h_0=1$ at the horizon. Comparing their constant terms in Eq.~\eqref{eq:green-proportional-selected} therefore gives
\begin{equation}
    c=1.
\end{equation}
It follows that
\begin{equation}
    H_0^{(N)}+\alpha_N^{(\gamma_1)}H_{\rm out}^{(N)}
    =
    H_0^{(N)}+\alpha_N^{(\gamma_2)}H_{\rm out}^{(N)},
\end{equation}
and hence
\begin{equation}
    \left(\alpha_N^{(\gamma_1)}-\alpha_N^{(\gamma_2)}\right)
    H_{\rm out}^{(N)}=0.
\end{equation}
	
Since $H_{\rm out}^{(N)}$ is a nonzero solution, this would require
\begin{equation}
    \alpha_N^{(\gamma_1)}=\alpha_N^{(\gamma_2)},
\end{equation}
contrary to the assumption. Therefore
\begin{equation}
    \boxed{
        \alpha_N^{(\gamma_1)}\neq\alpha_N^{(\gamma_2)}
        \quad\Longrightarrow\quad
        [A_{\gamma_1}:B_{\gamma_1}]
        \neq
        [A_{\gamma_2}:B_{\gamma_2}]
    }
    \label{eq:green-different-alpha-different-projective}
\end{equation}
	
Thus, if different curves select different Frobenius coefficients $\alpha_N^{(\gamma)}$ at the horizon, that difference cannot disappear under continuation to the boundary. Different values of $\alpha_N^{(\gamma)}$ select different projective classes of the source--response coefficient pair and therefore correspond to different limits of
\begin{equation}
    G_R^{(\gamma)}\propto\frac{B_\gamma}{A_\gamma}
\end{equation}
along the two parameter-space curves.

	\subsection{First-order directional ambiguity}
	\label{subsec:green-directional-ambiguity}
	
The simultaneous-zero condition
\begin{equation}
    \widehat A(P_*)=\widehat B(P_*)=0
\end{equation}
by itself does not imply that different first-order directions approaching the pole-skipping point must yield different Green-function limits. If the leading linear terms of the two regularized coefficients are proportional, the first-order limit of their ratio can be the same along different directions. This distinction between near-horizon degeneracy and boundary limiting behavior is also important in pole-skipping classifications \cite{AhnClassifying2021}.
	
Consider the simplest two-parameter slice. Fix the background and all other parameters and use $(\omega,k)$ as local coordinates. Near the pole-skipping point
\begin{equation}
    P_*=(\omega_N,k_*),
\end{equation}
write
\begin{align}
    \widehat A(\omega,k)
    &=a_\omega\,\delta\omega+a_k\,\delta k+O(\delta P^2),
    \label{eq:green-Ahat-linear}\\
    \widehat B(\omega,k)
    &=b_\omega\,\delta\omega+b_k\,\delta k+O(\delta P^2),
    \label{eq:green-Bhat-linear}
\end{align}
where
\begin{equation}
    \delta\omega:=\omega-\omega_N,
    \qquad
    \delta k:=k-k_*,
\end{equation}
and
\begin{align}
    a_\omega&:=(\partial_\omega\widehat A)_{P_*},
    &a_k&:=(\partial_k\widehat A)_{P_*},\\
    b_\omega&:=(\partial_\omega\widehat B)_{P_*},
    &b_k&:=(\partial_k\widehat B)_{P_*}.
\end{align}
	
Suppose a curve approaching $P_*$ satisfies to first order
\begin{equation}
    \delta k=v\,\delta\omega+o(\delta\omega),
    \label{eq:green-approach-slope}
\end{equation}
where $v$ is its first-order slope. Different values $v_1,v_2,\ldots$ represent different first-order approach directions.
	
Substituting Eq.~\eqref{eq:green-approach-slope} into Eqs.~\eqref{eq:green-Ahat-linear} and \eqref{eq:green-Bhat-linear} gives
\begin{align}
    \widehat A&=(a_\omega+va_k)\delta\omega+o(\delta\omega),\\
    \widehat B&=(b_\omega+vb_k)\delta\omega+o(\delta\omega).
\end{align}
Hence, along a direction for which $a_\omega+va_k\neq0$,
\begin{equation}
    \boxed{
        G_R^{(v)}\propto\frac{b_\omega+vb_k}{a_\omega+va_k}
    }
    \label{eq:green-slope-ratio}
\end{equation}
Thus a fixed slope $v$ determines a Green-function limit associated with that direction.
	
The remaining question is whether different slopes select different limits. Define the right-hand side of Eq.~\eqref{eq:green-slope-ratio} by
\begin{equation}
    F(v):=\frac{b_\omega+vb_k}{a_\omega+va_k}.
\end{equation}
Then
\begin{equation}
    \frac{dF}{dv}
    =
    \frac{a_\omega b_k-a_k b_\omega}{(a_\omega+va_k)^2}.
    \label{eq:green-slope-derivative}
\end{equation}
	
Therefore, if
\begin{equation}
    \boxed{
        \det
        \begin{pmatrix}
            \partial_\omega\widehat A & \partial_k\widehat A\\
            \partial_\omega\widehat B & \partial_k\widehat B
        \end{pmatrix}_{P_*}
        =a_\omega b_k-a_k b_\omega\neq0
    }
    \label{eq:green-jacobian-nondegenerate}
\end{equation}
then
\begin{equation}
    \frac{dF}{dv}\neq0
\end{equation}
wherever $a_\omega+va_k\neq0$. The Green-function limit therefore genuinely depends on the first-order approach slope $v$.
	
Equation~\eqref{eq:green-jacobian-nondegenerate} is thus a sufficient local nondegeneracy condition for different first-order approach directions to select different Green-function limits.
	
In particular, if for some slope $v=v_{\rm pole}$,
\begin{equation}
    a_\omega+v_{\rm pole}a_k=0
\end{equation}
while
\begin{equation}
    b_\omega+v_{\rm pole}b_k\neq0,
\end{equation}
then that direction gives a pole rather than a finite ratio,
\begin{equation}
    G_R^{(v_{\rm pole})}=\infty.
\end{equation}
Similarly, if
\begin{equation}
    b_\omega+v_{\rm zero}b_k=0,
    \qquad
    a_\omega+v_{\rm zero}a_k\neq0,
\end{equation}
then
\begin{equation}
    G_R^{(v_{\rm zero})}=0,
\end{equation}
corresponding to a zero.
	
Conversely, if
\begin{equation}
    a_\omega b_k-a_kb_\omega=0,
\end{equation}
then the leading linear forms of $\widehat A$ and $\widehat B$ may be proportional, and the first-order Green-function ratio need not depend on the approach slope. This does not exclude directional dependence at higher order. In that case the linear terms in Eqs.~\eqref{eq:green-Ahat-linear} and \eqref{eq:green-Bhat-linear} are insufficient, and quadratic or higher-order dependence on the parameters must be examined.
	
If the radial equation is analytic in
\begin{equation}
    q:=k^2
\end{equation}
rather than in $k$, particularly when $k_*=0$, then it is more natural to use $(\omega,q)$ as local spectral coordinates. All of the preceding formulas apply with $k$ replaced by $q$.

\subsection{Horizon-to-boundary structure of pole-skipping}
\label{subsec:green-summary}
	 
Combining the preceding results, the pole-skipping structure can be separated into three levels.
	 
First, in the horizon Frobenius recurrence,
\begin{equation}
    \boxed{D_N(P_*)=0}
\end{equation}
fixes the integer resonance.
	 
Second, at that resonance,
\begin{equation}
    \boxed{\Omega_N(P_*)=0}
\end{equation}
implies
\begin{equation}
    \det M_N=0,
    \qquad
    \ell_N=0,
    \qquad
    \dim\mathcal S_{\rm smooth}=2,
\end{equation}
and, simultaneously at the boundary,
\begin{equation}
    \boxed{\widehat A(P_*)=\widehat B(P_*)=0.}
\end{equation}
Thus the regularized source--response ratio is $0/0$ at the exact point.
	 
Finally, once an approach curve $\gamma$ is specified,
\begin{equation}
    \boxed{
        \alpha_N^{(\gamma)}
        =
        -\frac{
            \left.\dfrac{d}{ds}\Omega_N(\gamma(s))\right|_{s=0}
        }{
            \left.\dfrac{d}{ds}D_N(\gamma(s))\right|_{s=0}
        }
    }
\end{equation}
selects the actual nonzero resonant solution within the two-dimensional smooth horizon solution space, and continuation of that solution to the boundary determines the Green-function limit along the chosen parameter-space curve,
\begin{equation}
    \boxed{
        G_R^{(\gamma)}
        \propto
        \frac{B_0+\alpha_N^{(\gamma)}B_{\rm out}}{A_0+\alpha_N^{(\gamma)}A_{\rm out}}
    }.
\end{equation}
	 
The horizon-to-boundary dictionary established here therefore does more than place the rank loss of the horizon recurrence matrix next to a boundary $0/0$. The same solvability function $\Omega_N$ connects
	 
\begin{equation}
    \boxed{
        \begin{gathered}
            \text{resonant Frobenius solvability},\\
            \text{the two-dimensional smooth horizon solution space},\\
            \text{simultaneous vanishing of }\widehat A\text{ and }\widehat B,\\
            \text{and selection of the resonant Green-function limit}\\
            \text{along a parameter-space curve}
        \end{gathered}
    }
\end{equation}
within a single parameter-dependent structure.

	 \section{Discussion}
	 \label{sec:discussion}
	 
We have analyzed the relation between integer Frobenius resonance and the boundary $0/0$ structure of holographic pole-skipping for a general second-order scalar radial problem with a nonextremal horizon. The main goal was not to rederive the well-known local horizon degeneracy itself, but to formulate how that local degeneracy is transmitted to independently defined source and response channels at the boundary as a general horizon-to-boundary connection problem.

\subsection{Local resonance and the global horizon-to-boundary relation}
\label{subsec:discussion-local-global}
	 
At exact integer resonance,
\begin{equation}
    \nu=N,
    \qquad
    N\in\mathbb N,
\end{equation}
the resonant order of the horizon Frobenius recurrence degenerates because
\begin{equation}
    D_N=0.
\end{equation}
When the smaller-root solution is normalized by $h_0=1$, the solvability condition at that order is
\begin{equation}
    \Omega_N=0.
\end{equation}
	 
The local equivalence obtained in Section~\ref{sec:fuchsian} is
\begin{equation}
    \boxed{
        \det M_N=0
        \quad\Longleftrightarrow\quad
        \Omega_N=0
        \quad\Longleftrightarrow\quad
        \ell_N=0
        \quad\Longleftrightarrow\quad
        \dim\mathcal S_{\rm smooth}=2
    }
    \label{eq:discussion-local-equivalence}
\end{equation}
Integer-separated indicial roots, logarithmic Frobenius solutions, and solvability of the resonant recurrence are classical features of Fuchsian theory. The role of Eq.~\eqref{eq:discussion-local-equivalence} here is to connect explicitly, within a single normalization and notation, the recurrence-matrix rank loss commonly used in the pole-skipping literature with these classical Frobenius data.
	 
The main global step comes next. Away from resonance, writing
\begin{equation}
    H_{\rm in}=AR_1+BR_2,
\end{equation}
and assuming that the boundary basis is linearly independent, we obtained
\begin{equation}
    \boxed{
        A=0
        \quad\Longleftrightarrow\quad
        R_2\in C^\infty\text{ at the horizon}
    }
\end{equation}
and
\begin{equation}
    \boxed{
        B=0
        \quad\Longleftrightarrow\quad
        R_1\in C^\infty\text{ at the horizon}
    }.
\end{equation}
Thus the boundary conditions $A=0$ and $B=0$ are interpreted as the conditions under which the opposite boundary-normalized solution lies in the ingoing horizon-smooth direction.
	 
As resonance is approached, the nonresonant ingoing Frobenius coefficients can develop a simple parameter pole. Multiplying by the factor $g_N$ required by the recurrence to remove this pole and defining
\begin{equation}
    H_{\rm in}^{\rm reg}=g_NH_{\rm in},
\end{equation}
we obtain at exact resonance the identity
\begin{equation}
    \boxed{
        H_{\rm in}^{\rm reg}(x;N,\lambda)
        =
        \frac{u_N(N,\lambda)}{N}
        \Omega_N(N,\lambda)
        H_{\rm out}^{(N)}(x;\lambda)
    }
    \label{eq:discussion-exact-identity}
\end{equation}
	 
Continuing this identity to the boundary gives
\begin{align}
    \widehat A(N,\lambda)
    &=
    \frac{u_N(N,\lambda)}{N}\Omega_N(N,\lambda)A_{\rm out},\\
    \widehat B(N,\lambda)
    &=
    \frac{u_N(N,\lambda)}{N}\Omega_N(N,\lambda)B_{\rm out}.
\end{align}
Since $H_{\rm out}^{(N)}$ is nonzero and $R_1,R_2$ are independent, it follows that
\begin{equation}
    \boxed{
        \Omega_N(P_*)=0
        \quad\Longleftrightarrow\quad
        \widehat A(P_*)=\widehat B(P_*)=0
    }
    \label{eq:discussion-global-dictionary}
\end{equation}
	 
The local solvability function $\Omega_N$ therefore connects the degeneracy of the smooth horizon solution space and the simultaneous vanishing of the regularized source and response coefficients at the boundary within a single parameter-dependent connection problem, rather than leaving them as two unrelated phenomena.

	 \subsection{Exact \texorpdfstring{$0/0$}{0/0} and resonant limits along parameter-space curves}
	 \label{subsec:discussion-zero-path}
	 
A central interpretive point is to distinguish the two kinds of boundary coefficients that appear at a pole-skipping point.
	 
At the exact point $P_*$, if
\begin{equation}
    \Omega_N(P_*)=0,
\end{equation}
then
\begin{equation}
    H_{\rm in}^{\rm reg}(P_*)=0,
\end{equation}
and hence
\begin{equation}
    \boxed{
        (\widehat A(P_*),\widehat B(P_*))=(0,0)
    }.
\end{equation}
This pair consists of the boundary coefficients of the zero solution, and $[0:0]$ does not define a projective point in $\mathbb{CP}^1$. Thus the regularized source--response ratio is not determined pointwise at the exact pole-skipping point.
	 
By contrast, choose a parameter-space curve
\begin{equation}
    \gamma(s),
    \qquad
    \gamma(0)=P_*,
\end{equation}
and suppose that
\begin{equation}
    \left.\frac{d}{ds}D_N(\gamma(s))\right|_{s=0}\neq0.
\end{equation}
Then the limit of the nonresonant ingoing family normalized by $h_0=1$ is
\begin{equation}
    R_{\rm sel}^{(\gamma)}
    =
    H_0^{(N)}
    +
    \alpha_N^{(\gamma)}H_{\rm out}^{(N)},
\end{equation}
where
\begin{equation}
    \boxed{
        \alpha_N^{(\gamma)}
        =
        -\frac{
            \left.\dfrac{d}{ds}\Omega_N(\gamma(s))\right|_{s=0}
        }{
            \left.\dfrac{d}{ds}D_N(\gamma(s))\right|_{s=0}
        }
    }
    \label{eq:discussion-alpha}
\end{equation}
	 
Here $R_{\rm sel}^{(\gamma)}$ is not the zero solution but an actual nonzero resonant solution normalized by $h_0=1$. Writing at the boundary
\begin{equation}
    R_{\rm sel}^{(\gamma)}=A_\gamma R_1+B_\gamma R_2,
\end{equation}
we have
\begin{equation}
    (A_\gamma,B_\gamma)\neq(0,0),
\end{equation}
so the projective class
\begin{equation}
    [A_\gamma:B_\gamma]\in\mathbb{CP}^1
\end{equation}
is well defined.
	 
With the standard holographic normalization,
\begin{equation}
    G_R\propto\frac{B}{A},
\end{equation}
and therefore the Green-function limit selected by the curve $\gamma$ is
\begin{equation}
    \boxed{
        G_R^{(\gamma)}
        \propto
        \frac{B_\gamma}{A_\gamma}
        =
        \frac{B_0+\alpha_N^{(\gamma)}B_{\rm out}}{A_0+\alpha_N^{(\gamma)}A_{\rm out}}
    }
    \label{eq:discussion-green-path}
\end{equation}
	 
Consequently,
\begin{equation}
    (\widehat A,\widehat B)=(0,0)
\end{equation}
and
\begin{equation}
    (A_\gamma,B_\gamma)\neq(0,0)
\end{equation}
are not contradictory. The former is the exact-point value of the parameter-dependent family after multiplication by the factor that removes the simple parameter pole; the latter consists of the boundary coefficients of the actual resonant solution selected when the $0/0$ point is approached along a specified direction.
	 
This distinction gives a bulk solution-space interpretation of the direction dependence of pole-skipping. One may additionally impose the physical condition that the resonant solution contain no component coming from the second nonresonant Frobenius solution. Although $H_{\rm out}^{(N)}$ is smooth at exact resonance, it is obtained by taking that second nonresonant Frobenius solution to resonance. Under the additional condition one selects
\begin{equation}
    \alpha_N^{(\gamma)}=0.
\end{equation}
For a first-order nondegenerate curve this is equivalent to
\begin{equation}
    \left.d\Omega_N\right|_{P_*}(\dot\gamma_*)=0.
\end{equation}
Thus arbitrary approach curves may select different values of $\alpha_N^{(\gamma)}$ and different projective classes $[A_\gamma:B_\gamma]$ within the two-dimensional smooth resonant solution space, whereas the additional physical condition selects the solution with $\alpha_N^{(\gamma)}=0$. This selection is additional physical input and does not follow automatically from local Fuchsian solvability or horizon smoothness.

The simultaneous zero does not, by itself, guarantee nontrivial first-order direction dependence. For example, in an $(\omega,k)$ parameter slice, if
\begin{equation}
    \det
    \begin{pmatrix}
        \partial_\omega\widehat A & \partial_k\widehat A\\
        \partial_\omega\widehat B & \partial_k\widehat B
    \end{pmatrix}_{P_*}
    \neq0,
    \label{eq:discussion-direction-condition}
\end{equation}
then different first-order approach slopes select different Green-function limits. If this determinant vanishes, however, linear order alone is insufficient to decide the direction dependence and higher-order parameter dependence must be examined.

\subsection{Three distinct roles}
\label{subsec:discussion-three-levels}
	 
The preceding results are most simply summarized by distinguishing the roles of three quantities.
	 
First,
\begin{equation}
    \boxed{D_N(P_*)=0}
\end{equation}
determines the resonance at which the indicial roots differ by an integer.
	 
Second,
\begin{equation}
    \boxed{\Omega_N(P_*)=0}
\end{equation}
determines solvability of the resonant recurrence, removes the logarithmic term at exact resonance, and ensures the existence of two independent smooth horizon solutions. Through horizon-to-boundary continuation it simultaneously gives
\begin{equation}
    \widehat A(P_*)=\widehat B(P_*)=0.
\end{equation}
	 
Third,
\begin{equation}
    \boxed{
        \left.\frac{d}{ds}\Omega_N(\gamma(s))\right|_{s=0}
    }
\end{equation}
is not a new solvability condition at the exact point. Rather, it determines which solution in the already existing two-dimensional smooth resonant solution space is selected as the limit of the nonresonant ingoing family along the parameter-space curve $\gamma$.
	 
Thus the three roles may be summarized as
\begin{equation}
    \boxed{
        \begin{array}{rcl}
            D_N(P_*)=0
            &:&
            \text{resonance},
            \\[2mm]
            \Omega_N(P_*)=0
            &:&
            \text{two smooth resonant solutions and }\widehat A(P_*)=\widehat B(P_*)=0,
            \\[2mm]
            \left.d\Omega_N\right|_{P_*}(\dot\gamma_*)
            &:&
            \text{selection of a solution along the approach curve }\gamma
        \end{array}
    }
    \label{eq:discussion-three-level-summary}
\end{equation}
	 
This separation prevents one from conflating the solution-space degeneracy at exact resonance with the parameter-dependent limiting procedure.

\subsection{Scope and assumptions}
\label{subsec:discussion-scope}
	 
The argument does not rely on a particular hypergeometric equation or on a closed-form solution. Its essential inputs are the local Frobenius recurrence, parameter-dependent analytic continuation, and the existence of two independent boundary channels.
	 
More specifically, the analysis uses the following structures.
	 
First, the horizon is nonextremal, and the metric coefficients and radial-equation coefficients are assumed analytic near the horizon in ingoing Eddington--Finkelstein coordinates. Under these conditions the horizon is a regular singular point with indicial roots
\begin{equation}
    0,
    \qquad
    \nu=\frac{i\omega}{\kappa}.
\end{equation}
	 
Second, the radial problem is a second-order linear scalar equation. Hence the local and boundary solution spaces are both two-dimensional, and the connection problem between the two bases can be described by a $2\times2$ matrix.
	 
Third, the boundary must admit two independent asymptotic channels corresponding to source and response. The same argument applies at an irregular singular boundary provided one fixes an asymptotic ray or Stokes sector and the two independent asymptotic channels and their connection coefficients are well defined.
	 
Fourth, the global continuation data between the horizon and the boundary must be fixed. If additional singular points, branch cuts, or nontrivial monodromy occur in between, the continuation path and the required monodromy data must be included as part of the global problem.
	 
Fifth, near resonance the boundary basis and the horizon-to-boundary connection map must be chosen to depend analytically on the parameters. 	 
These assumptions restrict the domain of the theorem, but importantly they do not require exact solvability. The mechanism established here depends on the local Fuchsian structure and the global connection between solution bases rather than on a particular special-function representation.

\subsection{Limitations and extensions of the present framework}
\label{subsec:discussion-outlook}
	 
The results are restricted to second-order scalar problems satisfying the assumptions above. Several important generalizations require separate analysis.
	 
First, the present argument does not apply directly to an extremal horizon. We have assumed a nonextremal horizon with
\begin{equation}
    A_{\rm met}'(r_h)\neq0,
\end{equation}
which yields a finite surface gravity $\kappa$ and a regular-singular Frobenius structure. In an extremal limit the near-horizon singularity structure and frequency dependence can change, so the resonance mechanism based on
\begin{equation}
    D_n=n(n-\nu)
\end{equation}
should not be assumed without a new local analysis. Extremal pole-skipping must therefore be reconsidered from the singularity structure appropriate to the extremal horizon.
	 
Second, for coupled perturbations a single scalar recurrence and a single solvability function $\Omega_N$ may not describe the full resonant structure. When several field components are coupled, the recurrence becomes a matrix system and solvability at the resonant order must be generalized to the corresponding linear-algebra problem. The general viewpoint of connecting local rank loss with boundary-channel data through a global connection problem remains natural, but the scalar theorem proved here does not by itself establish the corresponding matrix statement.
	 
Third, for higher-order radial equations or problems with more than two independent boundary channels, the boundary coefficient classes may have to be described in a higher-dimensional projective space rather than in $\mathbb{CP}^1$. The structure of resonant limits along parameter-space curves and of the associated boundary ambiguity can then be richer than in the present problem with two independent boundary solutions.
	 
Fourth, if the first-order directional ambiguity is degenerate, so that the determinant in Eq.~\eqref{eq:discussion-direction-condition} vanishes, higher parameter jets of $\widehat A$ and $\widehat B$ become relevant. If a pole-skipping point cannot be classified by a linear intersection, the natural next step is to analyze its quadratic or higher-order zero structure.
	 
These generalizations lie beyond the scope of the present paper, but the current results make clear which quantities must be tracked: extract the solvability data from the local resonant recurrence, remove the parameter singularity of the parameter-dependent ingoing family, continue that family to the independent boundary solutions, and compare the resulting projective class of the boundary coefficient pair.

\subsection{Conclusion}
\label{subsec:discussion-conclusion}
	 
The central result of this work is an explicit connection between local horizon resonance and the boundary data of pole-skipping.
	 
At exact resonance we obtain the horizon-to-boundary dictionary
\begin{equation}
    \boxed{
        \small
        \det M_N=0
        \Longleftrightarrow
        \Omega_N=0
        \Longleftrightarrow
        \ell_N=0
        \Longleftrightarrow
        \dim\mathcal S_{\rm smooth}=2
        \Longleftrightarrow
        \widehat A=\widehat B=0
    }
    \label{eq:discussion-final-dictionary}
\end{equation}
	 
We have also distinguished the simultaneous zero at the exact point from the boundary coefficients of a nonzero solution. The former states that, after multiplication by the analytic factor that removes the simple parameter pole, the regularized ingoing family has the zero solution as its value at the exact point. By contrast, once an approach curve is chosen, the nonresonant ingoing family with the normalization $h_0=1$ has the nonzero resonant limit
\begin{equation}
    R_{\rm sel}^{(\gamma)}
    =
    H_0^{(N)}
    +
    \alpha_N^{(\gamma)}H_{\rm out}^{(N)},
\end{equation}
whose boundary coefficient pair defines
\begin{equation}
    [A_\gamma:B_\gamma]\in\mathbb{CP}^1.
\end{equation}
	 
The $0/0$ ambiguity of pole-skipping is therefore not separate from the local enlargement of the smooth horizon solution space at exact resonance. Under the stated global continuation assumptions, the same resonant solvability data connect the local horizon degeneracy, the simultaneous vanishing of the regularized boundary connection coefficients, and the resonant Green-function limit along an approach curve within a single structure.

\appendix
	
\section{Convergence of parameter-dependent Frobenius series}
\label{app:frobenius-convergence}
This appendix verifies that the parameter-dependent Frobenius series used in Section~\ref{sec:dictionary} converge in a common neighborhood of the horizon. In particular, we show that the coefficients of the regularized ingoing family extend analytically across resonance in parameter space and that the resulting series converges locally uniformly with respect to the parameters.
	
The purpose of this result is to promote the coefficientwise relation obtained in Theorem~\ref{thm:dictionary-resonance-identity} to an identity between actual analytic solutions.

	\subsection{Parameter-dependent Fuchs equation}
	\label{app:frobenius-setup}
	
Abbreviate the parameters by
\begin{equation}
    \zeta=(\nu,\lambda).
\end{equation}
Choose a sufficiently small complex parameter neighborhood $U$ of an integer resonance
\begin{equation}
    \zeta_*=(N,\lambda_*),
    \qquad N\in\mathbb N.
\end{equation}
Near the horizon, write the radial equation as
\begin{equation}
    L_\zeta R
    :=
    x^2R''+xp(x;\zeta)R'+q(x;\zeta)R=0.
    \label{eq:app-frobenius-equation}
\end{equation}
Here
\begin{equation}
    p(x;\zeta)=\sum_{j=0}^{\infty}p_j(\zeta)x^j,
    \qquad
    q(x;\zeta)=\sum_{j=0}^{\infty}q_j(\zeta)x^j,
    \label{eq:app-frobenius-pq}
\end{equation}
and we assume that for some $a>0$, both $p$ and $q$ are analytic in $x$ and $\zeta$ on
\begin{equation}
    |x|<a,
    \qquad
    \zeta\in U.
    \label{eq:app-common-domain}
\end{equation}
	
Keeping the standard Fuchs normalization used in Section~\ref{sec:dictionary},
\begin{equation}
    p_0(\zeta)=1-\nu,
    \qquad
    q_0(\zeta)=0.
\end{equation}
The recurrence for the exponent-zero Frobenius branch is therefore
\begin{equation}
    D_n(\nu)h_n
    +
    \sum_{m=0}^{n-1}K_{n,m}(\zeta)h_m
    =0,
    \label{eq:app-frobenius-recurrence}
\end{equation}
where
\begin{equation}
    D_n(\nu)=n(n-\nu),
    \label{eq:app-Dn}
\end{equation}
and
\begin{equation}
    K_{n,m}(\zeta)
    =
    mp_{n-m}(\zeta)+q_{n-m}(\zeta).
    \label{eq:app-Knm}
\end{equation}
	
Shrinking $U$ if necessary, assume
\begin{equation}
    |\nu-N|<\frac12.
    \label{eq:app-nu-neighborhood}
\end{equation}
Then the only positive-integer order at which the recurrence coefficient $D_n(\nu)$ can vanish in $U$ is $n=N$.

\subsection{Analyticity of the regularized coefficients}
\label{app:frobenius-coeff-analyticity}
	
Let the nonresonant exponent-zero solution be
\begin{equation}
    H_{\rm in}(x;\zeta)
    =
    \sum_{n=0}^{\infty}h_n(\zeta)x^n,
    \qquad h_0=1.
    \label{eq:app-Hin}
\end{equation}
	
For $1\leq n<N$,
\begin{equation}
    D_n(N)=n(n-N)\neq0,
\end{equation}
so successive solution of the recurrence shows that
\begin{equation}
    h_0,h_1,\ldots,h_{N-1}
\end{equation}
are analytic near $\zeta_*$.
	
At $n=N$,
\begin{equation}
    N(N-\nu)h_N+\Omega_N(\zeta)=0,
\end{equation}
so
\begin{equation}
    h_N(\zeta)
    =
    \frac{\Omega_N(\zeta)}{N(\nu-N)}.
    \label{eq:app-hN-pole}
\end{equation}
	
As in Section~\ref{subsec:dictionary-regularization}, choose an analytic factor
\begin{equation}
    g_N(\zeta)
    =
    (\nu-N)u_N(\zeta),
    \qquad
    u_N(\zeta_*)\neq0,
    \label{eq:app-gN}
\end{equation}
which removes the simple parameter pole, and define
\begin{equation}
    H_{\rm in}^{\rm reg}
    =
    g_NH_{\rm in}
    =
    \sum_{n=0}^{\infty}\widehat h_n(\zeta)x^n.
    \label{eq:app-Hreg}
\end{equation}
Thus
\begin{equation}
    \widehat h_n=g_Nh_n.
\end{equation}
	
\begin{proposition}
    \label{prop:app-coeff-analyticity}
    For every $n\geq0$, the coefficient $\widehat h_n(\zeta)$ can be defined analytically in a sufficiently small parameter neighborhood containing $\zeta_*$. 
\end{proposition}
	
\begin{proof}
For $n<N$, both $h_n$ and $g_N$ are analytic, so $\widehat h_n=g_Nh_n$ is analytic.
		
At $n=N$, Eq.~\eqref{eq:app-hN-pole} gives
\begin{equation}
    \widehat h_N
    =
    g_Nh_N
    =
    \frac{u_N(\zeta)}{N}\Omega_N(\zeta),
    \label{eq:app-hhatN}
\end{equation}
which is again analytic.
		
Now let $n>N$. By taking $U$ sufficiently small,
\begin{equation}
    D_n(\nu)\neq0,
    \qquad n>N,
    \qquad \zeta\in U.
\end{equation}
Multiplying Eq.~\eqref{eq:app-frobenius-recurrence} by $g_N$ gives
\begin{equation}
    D_n(\nu)\widehat h_n
    +
    \sum_{m=0}^{n-1}K_{n,m}(\zeta)\widehat h_m
    =0.
    \label{eq:app-hhat-recurrence}
\end{equation}
Therefore
\begin{equation}
    \widehat h_n
    =
    -\frac{1}{D_n(\nu)}
    \sum_{m=0}^{n-1}K_{n,m}(\zeta)\widehat h_m.
    \label{eq:app-hhat-rec-solved}
\end{equation}
If $\widehat h_0,\ldots,\widehat h_{n-1}$ are analytic, the right-hand side is analytic, and hence so is $\widehat h_n$. The result follows by induction.
\end{proof}
	
Thus, although the nonresonant coefficients $h_n$ themselves can have a parameter pole for $n\geq N$, the coefficients $\widehat h_n$ obtained after removing the simple parameter pole extend analytically across resonance.

\subsection{A common disk of convergence}
\label{app:frobenius-common-disk}
	
Proposition~\ref{prop:app-coeff-analyticity} shows that every coefficient $\widehat h_n(\zeta)$ extends analytically across resonance. We now show that the corresponding power series converges in a common disk in the $x$-plane.
	
\begin{proposition}[Local uniform convergence]
    \label{prop:app-uniform-convergence}
    After choosing a sufficiently small neighborhood $U$ of $\zeta_*$, there exists $r>0$ such that
    \begin{equation}
        H_{\rm in}^{\rm reg}(x;\zeta)
        =
        \sum_{n=0}^{\infty}\widehat h_n(\zeta)x^n
    \end{equation}
    is defined for
    \begin{equation}
        |x|<r,
        \qquad
        \zeta\in U.
    \end{equation}
    More precisely, for every compact subset $K\Subset U$ and every $0<r_0<r$, the series converges uniformly on
    \begin{equation}
        K\times\{|x|\leq r_0\}.
    \end{equation}
\end{proposition}
	
\begin{proof}
Fix $K\Subset U$ and choose
\begin{equation}
    0<a_1<a.
\end{equation}
Since $p$ and $q$ are continuous on $K\times\{|x|=a_1\}$,
\begin{equation}
    M_p:=\sup_{\substack{\zeta\in K\\|x|=a_1}}|p(x;\zeta)|<\infty,
    \qquad
    M_q:=\sup_{\substack{\zeta\in K\\|x|=a_1}}|q(x;\zeta)|<\infty.
\end{equation}
By the Cauchy coefficient estimate, for every $j\geq0$ and $\zeta\in K$,
\begin{equation}
    |p_j(\zeta)|\leq M_pa_1^{-j},
    \qquad
    |q_j(\zeta)|\leq M_qa_1^{-j}.
    \label{eq:app-Cauchy-bound}
\end{equation}
		
Since $K$ is compact,
\begin{equation}
    V:=\sup_{\zeta\in K}|\nu|<\infty.
\end{equation}
For sufficiently large $n$, for example
\begin{equation}
    n\geq n_0>2V,
\end{equation}
we have
\begin{align}
    |D_n(\nu)|
    &=n|n-\nu|\nonumber\\
    &\geq n(n-|\nu|)\nonumber\\
    &\geq\frac12n^2.
    \label{eq:app-D-lower-bound}
\end{align}
		
Choose
\begin{equation}
    0<r<a_1
\end{equation}
and set
\begin{equation}
    \theta:=\frac{r}{a_1}<1.
\end{equation}
Assume inductively that
\begin{equation}
    |\widehat h_m(\zeta)|\leq Cr^{-m}.
    \label{eq:app-majorant-assumption}
\end{equation}
		
Setting $k=n-m$, Eqs.~\eqref{eq:app-Knm} and \eqref{eq:app-Cauchy-bound} give
\begin{equation}
    |K_{n,m}(\zeta)|
    \leq
    (mM_p+M_q)a_1^{-(n-m)}.
\end{equation}
Therefore Eq.~\eqref{eq:app-hhat-rec-solved} implies
\begin{align}
    |\widehat h_n|
    &\leq
    \frac{1}{|D_n|}
    \sum_{m=0}^{n-1}|K_{n,m}|\,|\widehat h_m|\nonumber\\
    &\leq
    \frac{Cr^{-n}}{|D_n|}
    \sum_{k=1}^{n}\left[(n-k)M_p+M_q\right]\theta^k.
    \label{eq:app-majorant-step}
\end{align}
		
Since
\begin{equation}
    (n-k)M_p+M_q\leq nM_p+M_q,
\end{equation}
we obtain
\begin{align}
    \sum_{k=1}^{n}\left[(n-k)M_p+M_q\right]\theta^k
    &\leq(nM_p+M_q)\sum_{k=1}^{n}\theta^k\nonumber\\
    &\leq(nM_p+M_q)\frac{\theta}{1-\theta}.
\end{align}
Hence
\begin{equation}
    \sum_{k=1}^{n}\left[(n-k)M_p+M_q\right]\theta^k
    \leq An+B,
    \label{eq:app-sum-bound}
\end{equation}
where
\begin{equation}
    A=\frac{M_p\theta}{1-\theta},
    \qquad
    B=\frac{M_q\theta}{1-\theta}.
\end{equation}
Using Eq.~\eqref{eq:app-D-lower-bound}, for sufficiently large $n$,
\begin{equation}
    \frac{An+B}{|D_n|}\leq1.
\end{equation}
Therefore
\begin{equation}
    |\widehat h_n(\zeta)|\leq Cr^{-n}.
    \label{eq:app-final-majorant}
\end{equation}
		
The remaining finitely many coefficients with $n<n_0$ are bounded on $K$ by Proposition~\ref{prop:app-coeff-analyticity}. By choosing $C$ sufficiently large, Eq.~\eqref{eq:app-final-majorant} can therefore be made to hold for every $n\geq0$.
		
For any
\begin{equation}
    0<r_0<r,
\end{equation}
we then have
\begin{equation}
    |\widehat h_n(\zeta)x^n|
    \leq
    C\left(\frac{r_0}{r}\right)^n,
    \qquad
    |x|\leq r_0,
    \qquad
    \zeta\in K.
\end{equation}
The right-hand side is a geometric series, so the Weierstrass $M$-test implies that $\sum_n\widehat h_nx^n$ converges uniformly on $K\times\{|x|\leq r_0\}$.
\end{proof}

	\subsection{Analytic solution identity}
	\label{app:frobenius-analytic-identity}
	
By Propositions~\ref{prop:app-coeff-analyticity} and \ref{prop:app-uniform-convergence},
\begin{equation}
    H_{\rm in}^{\rm reg}(x;\zeta)
    =
    \sum_{n=0}^{\infty}\widehat h_n(\zeta)x^n
\end{equation}
defines a local solution analytic in both $x$ and $\zeta$ on a sufficiently small common disk. Each term $\widehat h_n(\zeta)x^n$ is analytic in $(x,\zeta)$, and the series converges locally uniformly on compact subsets by Proposition~\ref{prop:app-uniform-convergence}; hence its sum is analytic.
	
At exact resonance $\nu=N$, Section~\ref{sec:dictionary} gives
\begin{equation}
    \widehat h_0=\cdots=\widehat h_{N-1}=0
\end{equation}
and
\begin{equation}
    \widehat h_N
    =
    \frac{u_N(N,\lambda)}{N}\Omega_N(N,\lambda).
\end{equation}
All coefficients with $n>N$ are then uniquely determined from this $n=N$ datum by the resonant recurrence.
	
Normalize the larger-root Frobenius solution as
\begin{equation}
    H_{\rm out}^{(N)}(x;\lambda)
    =
    x^N\left(1+b_1x+b_2x^2+\cdots\right).
\end{equation}
The Taylor coefficients of the two sides then agree at every order. Since both series actually converge, uniqueness of Taylor series implies that
\begin{equation}
    \boxed{
        H_{\rm in}^{\rm reg}(x;N,\lambda)
        =
        \frac{u_N(N,\lambda)}{N}\Omega_N(N,\lambda)
        H_{\rm out}^{(N)}(x;\lambda)
    }
    \label{eq:app-final-resonance-identity}
\end{equation}
holds as an identity between analytic functions in a sufficiently small neighborhood of the horizon.
	
Thus the identity at exact resonance used in Section~\ref{sec:dictionary} is not merely a formal recurrence relation. It is an equality between actual local analytic solutions of the same second-order radial equation.

\backmatter

\bmhead{Funding}

This research was supported by the Basic Science Research Program through
the National Research Foundation of Korea (NRF), funded by the Ministry of
Education (Grant No. RS-2026-25560158).

\bibliography{references_32_audited}

@article{SonStarinets2002,
  author = {Son, Dam T. and Starinets, Andrei O.},
  title = {Minkowski-space correlators in {AdS/CFT} correspondence: recipe and applications},
  journal = {JHEP},
  volume = {09},
  pages = {042},
  year = {2002},
  doi = {10.1088/1126-6708/2002/09/042},
  eprint = {hep-th/0205051},
  archivePrefix = {arXiv}
}

@article{HerzogSon2003,
  author = {Herzog, Christopher P. and Son, Dam T.},
  title = {Schwinger--Keldysh propagators from {AdS/CFT} correspondence},
  journal = {JHEP},
  volume = {03},
  pages = {046},
  year = {2003},
  doi = {10.1088/1126-6708/2003/03/046},
  eprint = {hep-th/0212072},
  archivePrefix = {arXiv}
}

@article{SkenderisVanRees2008,
  author = {Skenderis, Kostas and van Rees, Balt C.},
  title = {Real-Time Gauge/Gravity Duality},
  journal = {Phys. Rev. Lett.},
  volume = {101},
  pages = {081601},
  year = {2008},
  doi = {10.1103/PhysRevLett.101.081601},
  eprint = {0805.0150},
  archivePrefix = {arXiv},
  primaryClass = {hep-th}
}

@article{VanRees2009,
  author = {van Rees, Balt C.},
  title = {Real-time gauge/gravity duality and ingoing boundary conditions},
  journal = {Nucl. Phys. B Proc. Suppl.},
  volume = {192--193},
  pages = {193--196},
  year = {2009},
  doi = {10.1016/j.nuclphysbps.2009.07.078},
  eprint = {0902.4010},
  archivePrefix = {arXiv},
  primaryClass = {hep-th}
}

@article{KovtunStarinets2005,
  author = {Kovtun, Pavel K. and Starinets, Andrei O.},
  title = {Quasinormal modes and holography},
  journal = {Phys. Rev. D},
  volume = {72},
  pages = {086009},
  year = {2005},
  doi = {10.1103/PhysRevD.72.086009},
  eprint = {hep-th/0506184},
  archivePrefix = {arXiv}
}

@article{Maldacena1998,
  author = {Maldacena, Juan M.},
  title = {The Large {N} limit of superconformal field theories and supergravity},
  journal = {Adv. Theor. Math. Phys.},
  volume = {2},
  pages = {231--252},
  year = {1998},
  doi = {10.4310/ATMP.1998.v2.n2.a1},
  eprint = {hep-th/9711200},
  archivePrefix = {arXiv}
}

@article{Witten1998,
  author = {Witten, Edward},
  title = {Anti-de Sitter space and holography},
  journal = {Adv. Theor. Math. Phys.},
  volume = {2},
  pages = {253--291},
  year = {1998},
  doi = {10.4310/ATMP.1998.v2.n2.a2},
  eprint = {hep-th/9802150},
  archivePrefix = {arXiv}
}

@article{GubserKlebanovPolyakov1998,
  author = {Gubser, Steven S. and Klebanov, Igor R. and Polyakov, Alexander M.},
  title = {Gauge theory correlators from non-critical string theory},
  journal = {Phys. Lett. B},
  volume = {428},
  pages = {105--114},
  year = {1998},
  doi = {10.1016/S0370-2693(98)00377-3},
  eprint = {hep-th/9802109},
  archivePrefix = {arXiv}
}

@article{GrozdanovSchalmScopelliti2018,
  author = {Grozdanov, Sa{\v{s}}o and Schalm, Koenraad and Scopelliti, Vincenzo},
  title = {Black Hole Scrambling from Hydrodynamics},
  journal = {Phys. Rev. Lett.},
  volume = {120},
  pages = {231601},
  year = {2018},
  doi = {10.1103/PhysRevLett.120.231601},
  eprint = {1710.00921},
  archivePrefix = {arXiv},
  primaryClass = {hep-th}
}

@article{BlakeDavisonGrozdanovLiu2018,
  author = {Blake, Mike and Davison, Richard A. and Grozdanov, Sa{\v{s}}o and Liu, Hong},
  title = {Many-body chaos and energy dynamics in holography},
  journal = {JHEP},
  volume = {10},
  pages = {035},
  year = {2018},
  doi = {10.1007/JHEP10(2018)035},
  eprint = {1809.01169},
  archivePrefix = {arXiv},
  primaryClass = {hep-th}
}

@article{BlakeDavisonVegh2020,
  author = {Blake, Mike and Davison, Richard A. and Vegh, David},
  title = {Horizon constraints on holographic Green's functions},
  journal = {JHEP},
  volume = {01},
  pages = {077},
  year = {2020},
  doi = {10.1007/JHEP01(2020)077},
  eprint = {1904.12883},
  archivePrefix = {arXiv},
  primaryClass = {hep-th}
}

@article{NatsuumeOkamura2019,
  author = {Natsuume, Makoto and Okamura, Takashi},
  title = {Nonuniqueness of Green's functions at special points},
  journal = {JHEP},
  volume = {12},
  pages = {139},
  year = {2019},
  doi = {10.1007/JHEP12(2019)139},
  eprint = {1905.12015},
  archivePrefix = {arXiv},
  primaryClass = {hep-th}
}

@article{NatsuumeOkamura2020,
  author = {Natsuume, Makoto and Okamura, Takashi},
  title = {Holographic chaos, pole-skipping, and regularity},
  journal = {Prog. Theor. Exp. Phys.},
  volume = {2020},
  pages = {013B07},
  year = {2020},
  doi = {10.1093/ptep/ptz155},
  eprint = {1905.12014},
  archivePrefix = {arXiv},
  primaryClass = {hep-th}
}

@article{NatsuumeOkamuraFinite2019,
  author = {Natsuume, Makoto and Okamura, Takashi},
  title = {Pole skipping with finite-coupling corrections},
  journal = {Phys. Rev. D},
  volume = {100},
  pages = {126012},
  year = {2019},
  doi = {10.1103/PhysRevD.100.126012},
  eprint = {1909.09168},
  archivePrefix = {arXiv},
  primaryClass = {hep-th}
}

@article{WuHigherCurvature2019,
  author = {Wu, Xing},
  title = {Higher curvature corrections to pole-skipping},
  journal = {JHEP},
  volume = {12},
  pages = {140},
  year = {2019},
  doi = {10.1007/JHEP12(2019)140},
  eprint = {1909.10223},
  archivePrefix = {arXiv},
  primaryClass = {hep-th}
}

@article{NatsuumeOkamuraZeroTemp2021,
  author = {Natsuume, Makoto and Okamura, Takashi},
  title = {Pole-skipping and zero temperature},
  journal = {Phys. Rev. D},
  volume = {103},
  pages = {066017},
  year = {2021},
  doi = {10.1103/PhysRevD.103.066017},
  eprint = {2011.10093},
  archivePrefix = {arXiv},
  primaryClass = {hep-th}
}

@article{ChoiMezeiSarosi2021,
  author = {Choi, Changha and Mezei, M{\'a}rk and S{\'a}rosi, G{\'a}bor},
  title = {Pole skipping away from maximal chaos},
  journal = {JHEP},
  volume = {02},
  pages = {207},
  year = {2021},
  doi = {10.1007/JHEP02(2021)207},
  eprint = {2010.08558},
  archivePrefix = {arXiv},
  primaryClass = {hep-th}
}

@article{CeplakRamdialVegh2020,
  author = {{\v{C}}eplak, Nejc and Ramdial, Kushala and Vegh, David},
  title = {Fermionic pole-skipping in holography},
  journal = {JHEP},
  volume = {07},
  pages = {203},
  year = {2020},
  doi = {10.1007/JHEP07(2020)203},
  eprint = {1910.02975},
  archivePrefix = {arXiv},
  primaryClass = {hep-th}
}

@article{CeplakVegh2021,
  author = {{\v{C}}eplak, Nejc and Vegh, David},
  title = {Pole-skipping and Rarita--Schwinger fields},
  journal = {Phys. Rev. D},
  volume = {103},
  pages = {106009},
  year = {2021},
  doi = {10.1103/PhysRevD.103.106009},
  eprint = {2101.01490},
  archivePrefix = {arXiv},
  primaryClass = {hep-th}
}

@article{WangWang2022,
  author = {Wang, Diandian and Wang, Zi-Yue},
  title = {Pole Skipping in Holographic Theories with Bosonic Fields},
  journal = {Phys. Rev. Lett.},
  volume = {129},
  pages = {231603},
  year = {2022},
  doi = {10.1103/PhysRevLett.129.231603},
  eprint = {2208.01047},
  archivePrefix = {arXiv},
  primaryClass = {hep-th}
}

@article{AhnHyperbolicScalarVector2020,
  author = {Ahn, Yongjun and Jahnke, Viktor and Jeong, Hyun-Sik and Kim, Keun-Young and Lee, Kyung-Sun and Nishida, Mitsuhiro},
  title = {Pole-skipping of scalar and vector fields in hyperbolic space: conformal blocks and holography},
  journal = {JHEP},
  volume = {09},
  pages = {111},
  year = {2020},
  doi = {10.1007/JHEP09(2020)111}
}

@article{AhnHyperbolicBlackHoles2019,
  author = {Ahn, Yongjun and Jahnke, Viktor and Jeong, Hyun-Sik and Kim, Keun-Young},
  title = {Scrambling in hyperbolic black holes: shock waves and pole-skipping},
  journal = {JHEP},
  volume = {10},
  pages = {257},
  year = {2019},
  doi = {10.1007/JHEP10(2019)257}
}

@article{YuanGe2021,
  author = {Yuan, Haiming and Ge, Xian-Hui},
  title = {Pole-skipping and hydrodynamic analysis in Lifshitz, {AdS}$_2$ and Rindler geometries},
  journal = {JHEP},
  volume = {06},
  pages = {165},
  year = {2021},
  doi = {10.1007/JHEP06(2021)165},
  eprint = {2012.15396},
  archivePrefix = {arXiv},
  primaryClass = {hep-th}
}

@article{Ramirez2021,
  author = {Ramirez, David M.},
  title = {Chaos and pole skipping in {CFT}$_2$},
  journal = {JHEP},
  volume = {12},
  pages = {006},
  year = {2021},
  doi = {10.1007/JHEP12(2021)006}
}

@article{YuanGeKimJiAhn2023,
  author = {Yuan, Haiming and Ge, Xian-Hui and Kim, Keun-Young and Ji, Chang-Woo and Ahn, Yongjun},
  title = {Pole-skipping points in 2D gravity and {SYK} model},
  journal = {JHEP},
  volume = {08},
  pages = {157},
  year = {2023},
  doi = {10.1007/JHEP08(2023)157},
  eprint = {2303.04801},
  archivePrefix = {arXiv},
  primaryClass = {hep-th}
}

@article{AhnClassifying2021,
  author = {Ahn, Yongjun and Jahnke, Viktor and Jeong, Hyun-Sik and Kim, Keun-Young and Lee, Kyung-Sun and Nishida, Mitsuhiro},
  title = {Classifying pole-skipping points},
  journal = {JHEP},
  volume = {03},
  pages = {175},
  year = {2021},
  doi = {10.1007/JHEP03(2021)175},
  eprint = {2010.16166},
  archivePrefix = {arXiv},
  primaryClass = {hep-th}
}

@article{NingWangWang2023,
  author = {Ning, Sirui and Wang, Diandian and Wang, Zi-Yue},
  title = {Pole skipping in holographic theories with gauge and fermionic fields},
  journal = {JHEP},
  volume = {12},
  pages = {084},
  year = {2023},
  doi = {10.1007/JHEP12(2023)084},
  eprint = {2308.08191},
  archivePrefix = {arXiv},
  primaryClass = {hep-th}
}

@article{ChuaHartmanWeng2026,
  author = {Chua, Wan Zhen and Hartman, Thomas and Weng, Wayne W.},
  title = {Replica manifolds, pole skipping, and the butterfly effect},
  journal = {JHEP},
  volume = {05},
  pages = {108},
  year = {2026},
  doi = {10.1007/JHEP05(2026)108},
  eprint = {2504.08139},
  archivePrefix = {arXiv},
  primaryClass = {hep-th}
}

@article{LuRanWu2026,
  author = {Lu, Zhenkang and Ran, Cheng and Wu, Shao-Feng},
  title = {Bulk Spacetime Encoding via Boundary Ambiguities},
  journal = {Phys. Rev. Lett.},
  volume = {136},
  pages = {061603},
  year = {2026},
  doi = {10.1103/gmd3-zt39}
}

@misc{ChounKim2026,
  author = {Choun, Yoon-Seok and Kim, Ki-Seok},
  title = {A Horizon-to-Boundary Dictionary Linking Smooth Horizon Continuation, Pole-Skipping, and {$SL(2,\mathbb{R})$} Lowest-Weight Structure},
  year = {2026},
  note = {arXiv:2607.03800 [hep-th]},
  eprint = {2607.03800},
  archivePrefix = {arXiv},
  primaryClass = {hep-th}
}

@book{Ince1956,
  author = {Ince, Edward L.},
  title = {Ordinary Differential Equations},
  publisher = {Dover Publications},
  address = {New York},
  year = {1956}
}

@book{CoddingtonLevinson1955,
  author = {Coddington, Earl A. and Levinson, Norman},
  title = {Theory of Ordinary Differential Equations},
  publisher = {McGraw--Hill},
  address = {New York},
  year = {1955}
}

\end{document}